\documentclass[11pt]{article}

\title{An algorithm for calculating the set of superhedging portfolios in markets with transaction costs}

\author{Andreas L{\"o}hne \thanks{Martin-Luther-Universit{\"a}t Halle-Wittenberg, NWF II, Departement of Mathematics, 06099 Halle (Saale), andreas.loehne@mathematik.uni-halle.de}, Birgit Rudloff \thanks{Princeton University, Department of Operations Research
and Financial Engineering, Princeton, NJ 08544, USA,
brudloff@princeton.edu, research supported by NSF award DMS-1007938.}}

\usepackage{amsmath, amssymb, color}
\usepackage{amsthm}
\usepackage{verbatim}
\usepackage{pst-node}           
\usepackage{pst-tree}           
\usepackage{epsfig}

\newtheorem{theorem}{Theorem}
\newtheorem{corollary}[theorem]{Corollary}
\newtheorem{remark1}[theorem]{Remark}
\newtheorem{lemma}[theorem]{Lemma}

\newtheorem{example1}[theorem]{Example}
\newenvironment{remark}{\begin{remark1}\rm}{\end{remark1}}
\newenvironment{example}{\begin{example1}\rm}{\end{example1}}

\numberwithin{equation}{section} \numberwithin{theorem}{section}

\newcommand{\capf}{{\rm cap \,}}
\newcommand{\clco}{{\overline{\rm co} \,}}

\newcommand{\of}[1]{\ensuremath{\left( #1 \right)}}

\newcommand{\cb}[1]{\ensuremath{ \left\{ #1 \right\} }}
\newcommand{\sqb}[1]{\ensuremath{ \left[ #1 \right] }}

\newcommand{\bs}{\backslash}

\newcommand{\vp}{\ensuremath{\varphi}}

\newcommand{\R}{\mathrm{I\negthinspace R}}
\newcommand{\N}{\mathrm{I\negthinspace N}}

\newcommand{\suc}{{\rm succ \,}}

\newcommand{\hypo}{{\rm hypo \,}}

\newcommand{\diag}{{\rm diag}}

\newcommand{\Int}{{\rm int\,}}

\renewcommand{\P}{\mathcal{P}}
\newcommand{\D}{\mathcal{D}}

\begin{document}

\maketitle

\begin{abstract}
We study the explicit calculation of the set of superhedging portfolios of contingent claims in a discrete-time market model for $d$ assets with proportional transaction costs. 
The set of superhedging portfolios can be obtained by a recursive construction involving set operations, going backward in the event tree. We reformulate the problem as a sequence of linear vector optimization problems and solve it by adapting known algorithms. The corresponding superhedging strategy can be obtained going forward in the tree.
Examples are given involving multiple correlated assets and basket options.
Furthermore, we relate existing algorithms for the calculation of the scalar superhedging price to the set-valued algorithm by a recent duality theory for vector optimization problems. The main contribution of the paper is to establish the connection to linear vector optimization, which allows to solve numerically multi-asset superhedging problems under transaction costs.
 \\[.2cm]
\\[.2cm]
\\[.2cm]
\end{abstract}

\section{Introduction}
\label{sec_Intro}
In this paper, a method is provided for an explicit calculation of the set of superhedging portfolios of contingent claims in a discrete-time market model for $d$ assets with proportional transaction costs when the underlying probability space is finite.
The set of superhedging portfolios in markets with transaction costs has been characterized under appropriate no arbitrage conditions using consistent price systems, the pendant to the density process of equivalent martingale measures in markets with transaction costs (\cite{Kabanov99,KabanovSafarian09,Schachermayer04}).
 Most algorithms so far concerned the calculation of the scalar superhedging price in markets with two assets. The scalar superhedging price is the smallest superhedging price if a specific asset is chosen as the num\'{e}raire. Often, strong assumptions on the size of the transaction costs or the type of contingent claims had to be imposed to provide an algorithm to calculate the scalar replication price and to ensure that this price coincides with the smallest superhedging price (see e.g. \cite{BensaidLesnePagesScheinkman92,BoyleVorst92,PerrakisLefoll97,Palmer01}). Roux, Tokarz, Zastawniak \cite{RouxZastawniak08}, Roux \cite{Roux08} were able to drop those assumptions and develop an algorithm to calculate the scalar superhedging price in the two asset case that is based on the dual representation of the scalar superhedging price deduced in Jouini, Kallal \cite{JouiniKallal95}. The scalar superhedging price is important to deduce price bounds. But only the set of superhedging portfolios provides full information if one wants to carry out a superhedging strategy starting from an initial portfolio that might contain other assets besides the num\'{e}raire asset.  Thus, in this paper we follow the num\'{e}raire-free approach of Kabanov \cite{Kabanov99}, Schachermayer \cite{Schachermayer04} to develop an algorithm to calculate the set of all superhedging portfolios.
An algorithmic approach to calculate the set of superhedging portfolios was provided for $d=2$ assets in Roux, Zastawniak \cite{RouxZastawniak09}.
In the working paper \cite{RouxZastawniak11}, which was developed independently of our work and uses a different approach, Roux, Zastawniak  present an extension of their recursive representation \cite{RouxZastawniak09} of the set of superhedging portfolios to $d$ assets. 
Thus, the more theoretical results like theorem~\ref{theorem recursive SHP}, theorem~\ref{theorem JK} and corollary~\ref{lemma alg scalar shp} of the present paper can also be found in \cite{RouxZastawniak11} ((3.2), (3.3),  (6.5) and lemma 5.5), even for American and Bermudan options. The connection to linear vector optimization (section~\ref{sec_reformulationLVOP} and \ref{subsec_scalar problem as LVOP}) and the usage of Benson's algorithm \cite{HamLoeRud13} presented here 
add a novel and practical side to the topic as they allow to control the error caused by numerical inaccuracy. With Benson's algorithm one can calculate approximations of the set of superhedging portfolios for a pre-specified error level which can significantly reduce computational costs and becomes important for larger problems and sets with many vertices.

This paper concerns three goals. First, an algorithm is presented to calculate the \textit{set} of all superhedging portfolios in the $d$ asset case, as well as an algorithm to calculate the superhedging strategy when starting from an initial portfolio vector in the set of superhedging portfolios. The $d$ asset case also allows to consider basket options. Secondly, we will show that the superhedging problem in markets with transaction costs leads to a sequence of linear vector optimization problems. This generalizes the well known fact that in frictionless markets, superhedging leads to a sequence of linear (scalar) optimization problems (see e.g. chapter~7.1 in \cite{FoellmerSchied04}). Thirdly, we will show how the above mentioned scalar algorithm of \cite{RouxZastawniak08,Roux08} can be related to the algorithm presented here by a recent duality theory for vector optimization problems called  {\em geometric duality} \cite{HeydeLoehne08}. It is notable, that even for the determination of the scalar superhedging price the calculation of the set of superhedging portfolios at intermediate nodes is necessary. Thus, the calculation of the set of superhedging portfolios is not more involved than the calculation of the scalar superhedging price. 

The first one to connect coherent set-valued risk measures with linear vector optimization problems was Hamel \cite{Hamel11}, where in section~8 the set-valued average value at risk was formulated as a linear vector optimization problem. We pick up on that idea, as the set of superhedging portfolios can be seen as a set-valued coherent risk measure (see \cite{HamelHeydeRudloff10}), but adopt to the fact that in our dynamic setting and since superhedging strategies are path dependent, one rather formulates the problem as a sequence of linear vector optimization problems in the spirit of dynamic programming (see also section~\ref{sec_Scalarization}). This is related to time consistency properties of dynamic set-valued risk measures \cite{FeiRud13}.

The paper is structured as follows. Section~\ref{sec_Preliminaries} reviews basic definitions and results about market models with proportional transaction costs and superhedging in those markets. In section~\ref{sec_Recursive representation} a recursive construction of the set of superhedging portfolios is deduced.
An algorithm to calculate a superhedging strategy is presented in section~\ref{subsec_strategy}. 
Section~\ref{sec_reformulationLVOP} reformulates the superhedging problem in terms of linear vector optimization.
Section~\ref{sec_Examples} presents examples ranging from multi-period binomial models to correlated trees for multiple assets and includes European call options as well as basket options. Section~\ref{sec_Scalarization} presents the dual representation of the scalar superhedging price of Jouini, Kallal \cite{JouiniKallal95} generalized to the $d$ asset case, and provides an algorithm to calculate the scalar superhedging price. 
The relationship between the scalar and the set-valued algorithm is deduced in section~\ref{subsec_scalar problem as LVOP} via duality for the corresponding linear vector optimization problems.


\section{Preliminaries}
\label{sec_Preliminaries}
Consider a financial market where $d$ assets can be traded over finite discrete time $t=0,1,\dots,T$. As stochastic base we fix a filtered finite probability space $(\Omega,\mathcal F, (\mathcal F_t)_{t=0}^T,P)$ with $|\Omega|=N\in\N$, $P(\omega)>0$ for all $\omega\in\Omega$ and the usual assumptions regarding the filtration, in particular, $\{\emptyset,\Omega\}=\mathcal F_0\subseteq\mathcal F_1\subseteq...\subseteq\mathcal F_T=\mathcal F$. 

A portfolio vector at time $t$ is an $\mathcal F_t$-measurable random vector $V_t \colon \Omega \to \R^d$. The values
$V_t\of{\omega}$ of portfolio vectors are given in physical units, i.e., the $i$-th component of the vector
$V_{t}\of{\omega}$ is the number of units of the $i$-th asset in
the portfolio at time $t$ for $i = 1, \ldots, d$. Thus, we do not fix a reference asset, like a currency or some other num\'{e}raire, and treat all assets equally as it was initiated by Kabanov \cite{Kabanov99}.

A market model is an $(\mathcal F_t)_{t=0}^T$ adapted process $(K_t)_{t=0}^T$ of solvency cones. A solvency cone is a polyhedral convex cone $K_t\of{\omega}$ with $\R^d_+ \subseteq K_t\of{\omega}\neq \R^d$ for all $\omega \in \Omega$ and all $t\in\{0,1,\dots,T\}$ and models the proportional frictions between the assets according
to the geometric model introduced in Kabanov \cite{Kabanov99}, see also \cite{Schachermayer04,KabanovSafarian09}.
To be more precise, let the terms of trade at time $t$ be modeled via an $\mathcal F_t$-measurable $d\times d$ bid-ask-matrix $\Pi_t$, as in definition~1.1 in \cite{Schachermayer04}. That is, $\Pi_t$ is a matrix-valued map $\omega\mapsto \Pi_t(\omega)$, denoting the bid and ask prices for the exchange between the $d$ assets. The entry $\pi^{ij}$ of $\Pi_t$ denotes the number of units of asset $i$ for which an agent can buy one unit of asset $j$ at time $t$, i.e., the pair $(\frac{1}{\pi^{ji}},\pi^{ij})$ denotes the bid- and ask-prices of the asset $j$ in terms of the asset $i$. Furthermore, $\Pi_t$ satisfies
\begin{align*}
 \pi^{ij}>0,\quad  1\leq i,j\leq d,\\
 \pi^{ii}=1,\quad  1\leq i\leq d,\\
 \pi^{ij}\leq \pi^{ik}\pi^{kj},\quad  1\leq i,j,k\leq d.
\end{align*}
The solvency cone $K_t(\omega)$ is spanned by the vectors $\pi^{ij}e^i-e^j$, $1\leq i,j\leq d$, and the unit vectors
$e^i$, $1\leq i\leq d$. 
Some of those vectors might be redundant.
Throughout the paper we assume
\begin{equation}
\label{assumption K_t}
  \R^d_+\setminus\cb{0} \subseteq \Int K_t\of{\omega},
\end{equation}
which is satisfied if the $d$ assets are liquid in the sense that the exchange rates between any two assets are positive and finite. This assumption ensures that any asset $i\in\{1,...,d\}$  can be used as a num\'{e}raire asset, and thus is useful whenever we compare our method to scalar procedures (as in section~\ref{sec_Scalarization}). However, it is not needed for theorems~\ref{theorem recursive SHP}, \ref{algo} and corollary~\ref{theorem SHstrategy}, but is used in theorem~\ref{th_strat1}.

The cone $K_t$ is called the solvency cone since it includes precisely those portfolios at time $t$ which can be exchanged into portfolios with only non-negative components (by trading at the prevailing bid-ask prices at time $t$). Thus, $K_t$ contains both the information about exchange rates and proportional transaction costs at time $t\in\{0,1,\dots,T\}$.

By $L^0_d(\mathcal{F}_t,\R^d) =L^0_d(\Omega, \mathcal{F}_t, P;\R^d)$ we denote the linear space of all $\mathcal{F}_t$-measurable
random vectors $X \colon \Omega \to \R^d$ for  $t=0,1,\dots,T$. We denote by $L^0_d(\mathcal{F}_t,D_t) =L^0_d\of{\Omega, \mathcal{F}_t, P;D_t}$ those $\mathcal{F}_t$-measurable random vectors that take $P$-a.s. values in $D_t$. Since we work with a finite probability space, we can identify $L^0_d(\mathcal{F}_T,\R^d)$ with $\R^{dN}$.
An $\R^d$-valued adapted process $(V_t)_{t=0}^T$ is called a self-financing portfolio
process for the market given by $(K_t)_{t=0}^T$ if for all $t\in\{0,\dots,T\}$ it holds $V_t - V_{t-1} \in -K_{t}$ $P$-a.s. with the convention $V_{-1}=0$.
We denote by $A_T \subseteq L^0_d(\mathcal{F}_T,\R^d)$ the set of random vectors $V_T \colon \Omega \to \R^d$, each being the value of a self-financing portfolio process at time $T$, i.e. $A_T$ is the set of superhedgeable claims
starting from initial endowment $0\in\R^d$ at time zero. As it easily follows from the definition of
self-financing portfolio processes, $A_T$ is a convex cone and
\[
    A_T=-L^0_d(\mathcal{F}_0,K_0)-L^0_d(\mathcal{F}_1,K_1)-...-L^0_d(\mathcal{F}_T,K_T).
\]
Note that $x_0+A_T$ is the set of terminal values of self-financing portfolio processes starting from the initial portfolio vector $V_{-1}=x_0\in\R^d$ at time $t=0$.

The fundamental theorem of asset pricing states that on a finite probability space a market given by $\of{K_t}_{t=0}^T$ satisfies the no arbitrage condition if and only if there exists a consistent
price system $\of{Z_t}_{t = 0}^T$  (see \cite{KabanovStricker01,KabanovRasonyiStricker02,Schachermayer04}).
The definitions are as follows.
The market given by $\of{K_t}_{t=0}^T$ is said to satisfy the no arbitrage
property (NA) if $A_T \cap L^0_d(\mathcal{F}_T,\R^d_+) = \cb{0}$,
see \cite{Schachermayer04}. This property is also known under the name of weak no arbitrage condition \cite{KabanovStricker01,KabanovRasonyiStricker02}, here we follow the notion of Schachermayer \cite{Schachermayer04}.
By $K^+_t$ we denote the set-valued mapping $\omega \mapsto K_t^+\of{\omega}$,
where $K_t^+\of{\omega}$ denotes the positive dual cone of the cone $K_t\of{\omega}$ for $\omega \in \Omega$ and $t\in\{0,1,\dots,T\}$. Thus,
\[
K_t^+\of{\omega} = \cb{v \in \R^d \colon \forall u \in K_t\of{\omega} \colon v^Tu \geq 0}.
\]
An $\R^d_+$-valued adapted process $Z = \of{Z_t}_{t = 0}^T$ is called a consistent price system for the market model $\of{K_t}_{t=0}^T$ if $Z$ is a martingale
under $P$ and for all $t\in\{0,1,\dots,T\}$ it holds $Z_t \in K^+_t\bs\cb{0}$
$P$-a.s.
Let us denote the set of all consistent price systems by $\mathbb{Z}$.
An initial portfolio vector $x_0\in\R^d$ allows to superhedge a claim $X\in L^0_d(\mathcal{F}_T,\R^d)$ if there exists a self-financing portfolio process $V_T\in A_T$ such that $x_0+V_T=X$ $P$-a.s., i.e., if $X\in x_0+A_T$. Note that the term 'self-financing' also allows the agents to 'throw away' non-negative quantities of the assets, thus the above condition $x_0+V_T=X$ $P$-a.s. really means superhedging and not necessarily perfect replication.
Let us denote by $SHP_0(X)$ the set of all those portfolio vectors $x_0\in\R^d$ at time $t=0$ that allow to superhedge the claim $X\in L^0_d(\mathcal{F}_T,\R^d)$, i.e.,
\begin{align}
    \label{primalRM}
    SHP_0(X)= \cb{x_0\in \R^d \colon X \in x_0 +A_T}.
\end{align}
Let us denote the halfspace with normal vector $w\in\R^d$ by $G(w)=\{x\in\R^d\colon 0\leq w^Tx\}$. Let us define the dual elements $(Q,w)$ by
\begin{align*}
\mathcal{W}^1=  \bigg\{\of{Q,w} \in
\mathcal{M}^P_{1,d}\times \R^d\bs\cb{0} \colon
 \forall t \in \{0,1,\dots,T\} \colon E\sqb{\diag\of{w}\frac{dQ}{dP} \Big| \mathcal{F}_t} \in
L^1_d(\mathcal{F}_t,K_t^+)\bigg\}.
\end{align*}
$\mathcal{M}^P_{1,d} = \mathcal{M}^P_{1,d}\of{\Omega, \mathcal{F}_T}$ denotes the set of
all vector probability measures with components being absolutely continuous with respect
to $P$, i.e. $Q_i \colon \mathcal{F}_T \to \sqb{0,1}$ is a probability measure on
$\of{\Omega,\mathcal{F}_T}$ such that $\frac{dQ_i}{dP} \in L^1(\mathcal{F}_T,\R)$ for $i = 1, \ldots, d$.
The set $\mathcal{W}^1$ is one-to-one with the set of consistent price systems $\mathbb{Z}$ as it was shown in \cite{HamelHeydeRudloff10}.
The set of superhedging portfolios of a claim $X \in L^0_d(\mathcal{F}_T,\R^d)$ has the following dual representation.

\begin{theorem}[\cite{KabanovStricker01}, \cite{HamelHeydeRudloff10} corollary~5.4]
\label{ThmSH}
Under the no arbitrage condition (NA), for a claim $X\in L^0_d(\mathcal{F}_T,\R^d)$ one has
\begin{align}
    \label{dualRMSchachi}
    SHP_0(X)&= \cb{ x_0\in \R^d \colon \forall Z\in\mathbb{Z}:\;E[X^TZ_T]\leq x_0^TZ_0}
    \\
    \label{dualRM}
    &= \bigcap_{\of{Q, w} \in \mathcal{W}^1 } \of{E^Q\sqb{X} +G\of{w}}.
\end{align}
\end{theorem}
The function $X\mapsto SHP_0(-X)$ defines a closed set-valued coherent
market-compatible risk measure as introduced in \cite{HamelHeydeRudloff10}.
Equation \eqref{primalRM} is the primal representation of this risk measure and corresponds to the very definition of elements $x_0\in \R^d$ allowing to superhedge the claim $X\in L^0_d(\mathcal{F}_T,\R^d)$. Equation \eqref{dualRMSchachi} refers to the representation of superhedging portfolios with help of consistent price systems as in \cite{Kabanov99,KabanovStricker01,KabanovRasonyiStricker02,KabanovSafarian09,Schachermayer04}. 
Equation \eqref{dualRM} corresponds to the dual representation with help of vector probability measures which is closer to the scalar results and is in the spirit of set-valued coherent risk measure (see \cite{HamelHeydeRudloff10}). 

The {\em recession cone} of a convex set $A\subseteq \R^d$ is the set $A_\infty:=\{x\in\R^d: \forall t>0\;\; A+tx\subseteq A\}$, see e.g. \cite{Rockafellar97}.

\section{Recursive representation of the set of superhedging portfolios}
\label{sec_Recursive representation}

We will show that the set of superhedging portfolios $SHP_0(X)$ as defined in \eqref{primalRM} can be obtained by a recursive construction, going backward in the event tree, while the corresponding superhedging strategy can be obtained going forward in the tree. Based on this recursive structure, an algorithm to calculate superhedging portfolios and strategies will be developed.

Let $\Omega_t$ denote the set of atoms of $\mathcal F_t$ for $t\in\{0,1,\dots,T\}$. Let us consider a tree that represents the income of information, i.e., the nodes of the tree correspond to the atoms $\omega \in\Omega_t$ of $\mathcal F_t$. In the examples in section~\ref{sec_Examples} we will consider recombining trees such that the number of nodes does not grow exponentially with the number of time steps to ensure that the algorithm is computationally manageable. However, we do not need this assumption now.

A node $\bar\omega\in\Omega_{t+1}$ is called a successor node of $\omega\in\Omega_{t}$ $(t\in\{0,\dots,T-1\})$ if $\bar\omega\subseteq\omega$. The set of successor nodes of $\omega\in\Omega_{t}$ is denoted by
\[
	\suc (\omega)=\{\bar\omega\in\Omega_{t+1}:\bar\omega\subseteq\omega\}.
\]
Let the market model be described by an adapted stochastic process for the solvency cones $(K_t)_{t=0}^T$. The superhedging portfolios of a random payoff $X\in L^0_d(\mathcal{F}_T,\R^d)$ can be written in recursive form. 

\begin{theorem}
\label{theorem recursive SHP}
If the no arbitrage condition (NA) holds true, the set of superhedging portfolios $SHP_0(X)\subseteq\R^d$, defined in \eqref{primalRM}, of a claim $X\in L^0_d(\mathcal{F}_T,\R^d)$ satisfies
\begin{equation}\label{autokaputt}
 \emptyset \neq SHP_0(X) \neq \R^d
\end{equation}
and can be obtained recursively via
\begin{align}
\label{terminal condition SHP}
\forall\omega\in\Omega_T:\;\; &SHP_T(X)(\omega)=X(\omega)+K_T(\omega)
\\
\label{eqrecursive}
\forall t\in\{T-1,\dots,1,0\}, \forall\omega\in\Omega_{t}:\;\; &SHP_{t}(X)(\omega)=\bigcap_{\bar\omega\in\suc(\omega)} SHP_{t+1}(X)(\bar\omega)+K_t(\omega).
\end{align}
\end{theorem}

\begin{proof} 
Note that the intersection of finitely many non-empty polyhedral convex sets $A_1,\dots,A_m$ satisfying $(A_i)_\infty \supseteq \R^d_+$ for $i=1,\dots,m$ is non-empty (e.g. the component-wise maximum over a selection $a^i \in A_i$ $(i=1,\dots,m)$ belongs to the intersection). Since $K_t(\omega) \supseteq \R^d_+$ for all $t\in\cb{0,\dots,T}$ and all $\omega\in \Omega_t$, the sets $SHP_t(X)(\omega)$ defined by
\eqref{terminal condition SHP}, \eqref{eqrecursive} are non-empty.

Let us first show that the set $SHP_0(X)$ as defined by the recursive construction \eqref{terminal condition SHP}, \eqref{eqrecursive} is included in the set defined by \eqref{primalRM}.
Let $x_0\in SHP_0(X)$ as defined by the recursive construction \eqref{terminal condition SHP}, \eqref{eqrecursive}. Equation \eqref{eqrecursive} for $t=0$ is equivalent to the existence of a $V_0\in x_0-K_0$ such that $V_0\in\bigcap_{\omega\in\Omega_1} SHP_{1}(X)(\omega)$. In particular, $V_0\in SHP_{1}(X)(\omega)$ for all $\omega\in\Omega_1$. Using the definition of $SHP_{1}(X)(\omega)$, i.e. \eqref{eqrecursive} for $t=1$, there exists for each $\omega\in\Omega_{1}$ a $V_1(\omega)\in V_0-K_1(\omega)$ such that $V_1(\omega)\in SHP_{2}(X)(\bar\omega)$ for all $\bar\omega\in\suc (\omega)$. Going forward in the tree in the same manner and using \eqref{terminal condition SHP}, one can conclude that there exists for each $\bar\omega\in\Omega_{T}$  with $\bar\omega\in\suc (\omega)$ for some $\omega\in\Omega_{T-1}$ a $V_T(\bar\omega)\in V_{T-1}(\omega)-K_T(\bar\omega)$ such that $V_T(\bar\omega)=X(\bar\omega)$. That means, there exists a self-financing portfolio process $V_0,...,V_T$ with $V_T\in x_0+A_T$ and $V_T(\bar\omega)=X(\bar\omega)$, i.e. $X\in x_0+A_T$, which means $x_0$ is a superhedging portfolio of $X$ by \eqref{primalRM}.

For the reverse inclusion, take $x_0\in SHP_0(X)$ as defined in \eqref{primalRM}.
For a given path $(\omega_0,\dots,\omega_{T-1})$ with $\omega_t \in\Omega_{t}$ and $\omega_t\in\suc (\omega_{t-1})$ $(t=1,\dots,T-1)$, there exist $k_t\in K_t(\omega_t)$ such that
$x_0-k_0-...-k_{T-1}\in X(\bar\omega)+K_T(\bar\omega)$
for all $\bar\omega \in\Omega_{T}$ with $\bar\omega\in\suc (\omega_{T-1})$. Thus,
\[
    x_0-k_0-...-k_{T-2}\in \bigcap_{\bar\omega\in\suc(\omega_{T-1})} SHP_{T}(X)(\bar\omega)+k_{T-1},
\]
where $SHP_{T}(X)$ is defined as in \eqref{terminal condition SHP}. Since this holds for any $\omega_{T-1} \in\Omega_{T-1}$ with $\omega_{T-1}\in\suc (\omega_{T-2})$, we have
\[
    x_0-k_0-...-k_{T-2}\in \bigcap_{\bar\omega\in\suc(\omega)} SHP_{T}(X)(\bar\omega)+K_{T-1}(\omega)=SHP_{T-1}(X)(\omega)
\]
via \eqref{eqrecursive} for all $\omega \in\Omega_{T-1}$ with $\omega\in\suc (\omega_{T-2})$.
Proceeding like this reveals $x_0\in SHP_0(X)$ as defined by the recursive construction \eqref{terminal condition SHP}, \eqref{eqrecursive}.
This proves the equivalence of \eqref{primalRM} and  the recursive definition of $SHP_0(X)$ in \eqref{terminal condition SHP}, \eqref{eqrecursive}.
By the fundamental theorem of asset pricing, no arbitrage implies the existence of a consistent
price system $Z\in\mathbb{Z}$  (see \cite{KabanovStricker01}). Thus, $SHP_0(X)\neq \R^d$ follows from theorem~\ref{ThmSH}.
\end{proof}

\begin{remark}
The equivalence of \eqref{primalRM} and  the recursive definition of $SHP_0(X)$ in \eqref{terminal condition SHP}, \eqref{eqrecursive} also holds for probability spaces that are not necessarily finite, see remark 12 in \cite{FeiRud13} for a short and elegant proof using the recently developed multi-portfolio time consistency concept for dynamic set-valued risk measures and its equivalent characterization through recursiveness. For ease of notation one can write
\begin{align*}
&SHP_T(X)=X+K_T
\\
\forall t\in\{T-1,\dots,1,0\}:\;\; &SHP_{t}(X)=SHP_{t+1}(X)\cap L^0_d(\mathcal{F}_t,\R^d)+K_t.
\end{align*}
\end{remark}

Clearly, $SHP_t(X)(\omega)$ is the set of superhedging portfolios of $X\in L^0_d(\mathcal{F}_T,\R^d)$ at time $t$ at node $\omega\in\Omega_t$, $t \in\{0,1,...,T\}$ and \eqref{autokaputt} is satisfied likewise. As a consequence of theorem~\ref{theorem recursive SHP}, for all $t \in\{0,1,...,T\}$ and all $\omega\in\Omega_t$, $SHP_t(X)(\omega)$ is a non-empty polyhedral convex set. Under assumption~\eqref{assumption K_t} it satisfies
\begin{equation}\label{eq_rec}
\Int \of{SHP_t(X)(\omega)}_\infty \supseteq \R^d_+\setminus\cb{0}.
\end{equation}
The recursive structure \eqref{terminal condition SHP}, \eqref{eqrecursive} provides a geometric intuition for designing an algorithm to calculate the superhedging portfolios going backwards in the tree.
The operations involved are intersections of  polyhedra and the sum of polyhedral sets and polyhedral cones. These operations could be realized by methods from computational geometry, which are essentially based on the vertex enumeration problem, see e.g. \cite{BarDobHuh96,BreFukMar98}. For several reasons, discussed at the end of section \ref{sec_reformulationLVOP}, we reformulate the problem as a sequence of linear vector optimization problems. Note that for path independent payoffs and recombining trees, the recursive pricing procedure described in theorem~\ref{theorem recursive SHP}
grows only polynomial as the trading frequency increases.

Note that the same geometric intuition as in \eqref{terminal condition SHP}, \eqref{eqrecursive} appears if one is only interested in calculating the scalar superhedging price (i.e. the smallest superhedging price in a given currency or num\'{e}raire) when transaction costs are present. Even in the scalar case, it cannot be avoided to use the set-valued operations (intersection and sum of polyhedral sets). As a consequence of \eqref{eq_rec}, the polyhedral convex sets $SHP_t(X)(\omega)$ can be expressed as the epigraphs of piecewise linear functions $f:\R^{d-1}\rightarrow\R$. For example, in the two asset binomial model, this recursive structure can be rediscovered in the sequential problem $\mathcal Q_t$, p. 71 in \cite{BensaidLesnePagesScheinkman92}.

On the other hand, for the purpose of comparing our method with algorithms developed for calculating the scalar superhedging price based on a dual description as in \cite{RouxZastawniak08,RouxZastawniak09,Roux08} it is quite helpful to reformulate \eqref{terminal condition SHP}, \eqref{eqrecursive} as a sequence of linear vector optimization problems. As we will see in detail in section~\ref{sec_reformulationLVOP}, in each iteration step a solution of the dual vector optimization problem \cite{HeydeLoehne08} is used. As it will turn out, duality provides a link between \eqref{terminal condition SHP}, \eqref{eqrecursive} and above mentioned scalar algorithms and allows to recover the whole set of superhedging portfolios from intermediate results of the scalar algorithm. It is a somewhat surprising insight that the calculation of the set of superhedging portfolios is not more difficult than calculating the scalar superhedging price in markets with transaction costs, see section~\ref{sec_Scalarization}.

\begin{remark}
Note that it is sufficient to develop an algorithm for superhedging portfolios and superhedging strategies since the set of subhedging portfolios $SubHP_0(X)$ of a claim $X\in L^0_d(\mathcal{F}_T,\R^d)$ is just the negative of the set of superhedging portfolios of $-X$
 \[
    SubHP_0(X)=-SHP_0(-X).
 \]
The corresponding strategy for the buyer of a claim $X$ is to superhedge $-X$.
\end{remark}

\subsection{Calculation of the superhedging strategy}
\label{subsec_strategy}
From the proof of theorem~\ref{theorem recursive SHP}, one can see that the superhedging strategy, when starting from a particular element in the set of superhedging portfolios, can be calculated going forward in the tree. In the presence of transaction costs optimal superhedging strategies are, in general, path-dependent even for path-independent payoffs. But one only needs to compute the superhedging strategy along the realized path, one step at a time, in real time.

\begin{remark}
A self-financing trading strategy starting from an initial portfolio vector $x_0\in\R^d$ is a sequence $k_{0}, k_{1},...,k_T$ of trades at time $t=0,1,\dots,T$ with $k_{t}\in K_{t}$ for all $t$. The $j$-th component of $k_{t}$ gives the number of shares of asset $j$ to be bought/sold at time $t$ and paying the prevailing transaction costs as given by $K_{t}$. Of course, such a self-financing trading strategy is one-to-one with a self-financing portfolio process $V_0, V_{1},...,V_T$ via
\[
 V_0=x_0-k_0 \quad\text{ and }\quad \forall t\in\{1,\dots,T\} \colon\quad V_{t} = V_{t-1} -k_{t}
\]
as it can be seen from the definition of a self-financing portfolio process.
\end{remark}

\begin{corollary}
\label{theorem SHstrategy}
For a claim $X\in L^0_d(\mathcal{F}_T,\R^d)$ and a given initial portfolio vector $x_0\in\R^d$ with $x_0\in SHP_0(X)$, there exists a superhedging strategy for a path $(\omega_0,\dots,\omega_T)$ with $\omega_t \in\Omega_{t}$ and $\omega_t\in\suc (\omega_{t-1})$ $(t=1,\dots,T)$. Such a strategy is given by a self-financing portfolio process $V_0, V_{1},...,V_T$ satisfying
\begin{align}
\label{SHstrategy0}
V_0&\in\of{\cb{x_0}-K_0}\cap \bigcap_{\bar\omega\in\suc (\omega_0)} SHP_{1}(X)(\bar\omega),
\\
\label{SHstrategy}
	V_{t}&\in\of{\cb{V_{t-1}}-K_{t}(\omega_t)}\cap
	\bigcap_{\bar\omega\in\suc (\omega_t)} SHP_{t+1}(X)(\bar\omega),
\end{align}
for all $t\in\{1,...,T\}$.
\end{corollary}

\begin{proof}
The assertions follow from the recursion \eqref{terminal condition SHP}, \eqref{eqrecursive} and the proof of theorem~\ref{theorem recursive SHP}. In particular, it follows that the sets in \eqref{SHstrategy0}, \eqref{SHstrategy} are non-empty and thus the existence of a superhedging strategy follows.
\end{proof}

Note that a superhedging strategy is typically not uniquely determined. 
In the following we will shortly discuss how to choose one specific superhedging strategy. Since the strategies are in general superhedging and not replication strategies, it might be possible in certain scenarios to withdraw cash or assets without endangering the superhedging criteria. Thus, one possible criterion to choose a strategy might be to withdraw as much as possible of certain assets at intermediate points in time. Different criterions are possible and are discussed in detail in \cite{LoehneRudloff12OR}.

At each time $t=0,1,...,T$, choosing a superhedging strategy at a node $\omega\in\Omega_{t}$ with endowment $x=V_{t}(\omega)\in SHP_{t}(X)(\omega)$ means choosing a triplet $(v,y,k)$, where $v=V_{t+1}\in SHP_{t+1}(X)(\omega)$ is the portfolio after a trade $k\in K_{t}(\omega)$ and withdrawals of a portfolio $y\in\R^d$ at time $t$. That means, $V_{t+1}=V_{t}-y-k$ and the portfolio $V_{t+1}$ is hold from time $t$ to $t+1$ and presents the initial endowment (before trades are made) in the next iteration step.

Let us assume an investor following a superhedging strategy wants to withdraw as much of a certain portfolio $y\in\R^d_+\setminus\cb{0}$ as possible at each intermediate point in time.

Let $t \in \cb{1,\dots,T}$, $\omega \in \Omega_{t}$, $x \in SHP_{t}(X)(\omega)$ and let $\widetilde{K}_t(\omega)\in \R^{d \times s}$ be a matrix containing the $s$ generating vectors of $K_{t}(\omega)$. We assume that inequality representations of the polyhedral sets $SHP_{t+1}(X)(\omega)$ are known. An algorithm to compute them will be given in section \ref{sec_reformulationLVOP}. Solving the following LP with variables $(v,\alpha,z)\in \R^d \times \R \times \R^s$, we obtain a portfolio $\bar v$ according to corollary~\ref{theorem SHstrategy} which has the property that a maximal amount of the portfolio $y$ is withdrawn at time $t$.
If $t<T$, we take
\begin{equation}\label{eq_strat1a}
\max \alpha \quad \text{ s.t. } \; v \in  \bigcap_{\bar\omega \in \suc(\omega)} SHP_{t+1}(X)(\bar\omega),\; v + \alpha y + \widetilde{K}_t(\omega) z = x,\;  \alpha\geq 0\; ,z \geq 0.
\end{equation}
At time $T$, we solve
\begin{equation}\label{eq_strat1b}
 \max \alpha \quad \text{ s.t. } \;v \geq X(\omega),\; v + \alpha y + \widetilde{K}_t(\omega) z  = x,\;  \alpha\geq 0\; ,z \geq 0.
\end{equation}

From a solution $(\bar v,\bar\alpha,\bar z)$, we get the new portfolio $V_{t+1}=\bar v$. The portfolio $\bar y =\bar\alpha y\in \R^d_+$ describes the withdrawal and $k=\widetilde{K}_t(\omega) \bar z$ is the corresponding trade.

\begin{theorem}\label{th_strat1}
If the market satisfies the no arbitrage
property (NA), then there exists a solution $(\bar v,\bar\alpha,\bar z)$ for \eqref{eq_strat1a} and \eqref{eq_strat1b}.
\end{theorem}
\begin{proof}
The feasible set is non-empty by corollary~\ref{theorem SHstrategy} and the fact that any nonnegative withdrawal $y$ belongs to $K_{t}(\omega)$. Assume there does not exist a solution $(\bar v,\bar\alpha,\bar z)$ for \eqref{eq_strat1a} or \eqref{eq_strat1b}, i.e. the value of the problem is $+\infty$. Denoting by $C=(SHP_t(X)(\omega))_\infty$ the recession cone of $SHP_t(X)(\omega)$, we obtain $y \in -C$. But $y \in \Int C$ by \eqref{eq_rec}. Hence $0 \in \Int C$ and thus $C=\R^d$. It follows that $SHP_t(X)(\omega)=\R^d$, which contradicts theorem~\ref{theorem recursive SHP}.
\end{proof}

An important special case is the {\em max-cash superhedging strategy}, which is obtained by setting $y=(1,0,\dots,0)^T$, where the first component is assumed to correspond to the cash account of interest.

\begin{example}
\label{example d=2, n=2} Let us consider a simple introductory example: a one period binomial model with non-constant proportional transaction costs, where the set of superhedging portfolios has multiple vertices. We will use this example to illustrate the algorithm and to explain differences between the scalar and the set-valued approach. Note that the transaction costs are chosen to be quite large, just for the purpose of obtaining illustrative pictures.

Let asset $0$ be a riskless cash account and let us assume for simplicity that interest rates are zero. Asset $1$ is a risky stock, whose bid-ask prices $(S_t^b,S_t^a)$ at time $t=0$ and $t=T$ are modeled as follows:
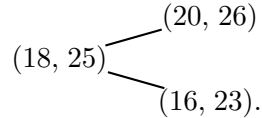
\begin{center}
\pstree[treemode=R]{\TR{(18,\;25)}}{\TR{(20,\;26)}\TR{(16,\;23).}}
\end{center}
We consider a digital option, more specifically an asset or nothing call option with physical delivery and strike $K=24$.
The payoff is given by
\[
X\of{\omega} = \of{X_1\of{\omega}, X_2\of{\omega}}^T = (0, I_{\cb{S_T^a \geq
K}}\of{\omega})^T.
\]
Thus, the payoff in the up-node is $X(\omega_1)=(0, 1)^T$ and in the down node
$X(\omega_2)=(0, 0)^T$. The calculation of the set of superhedging portfolios by the recursive procedure described in theorem~\ref{theorem recursive SHP} and illustrated in figure~\ref{ex_toy3} reveals that $SHP_0(X)$ has two vertices, one at $(0,1)^T$ and one at $(-80,5)^T$ and a recession cone equal to the solvency cone $K_0$ at
initial time which is generated by $\of{-18, 1}^T$ and $\of{25, -1}^T$. Figure~\ref{ex_toy2} shows the set of sub- and superhedging portfolios.
The scalar superhedging price is given by $25$ units cash and corresponds to the buy and hold strategy that superreplicates the claim. The strategy is to transfer the initial position $(25,0)^T$ into the vertex $(0,1)^T$ of $SHP_0(X)$ at initial time and hold this portfolio until terminal time.
However, the knowledge of the scalar superhedging price and the corresponding strategy
does not give information about optimal strategies if one already owns some shares of the
stock. For example, if one owns $5$ stocks and is short $80$ units cash at initial time
(the portfolio corresponding to the second vertex of $SHP_0(X)$), one cannot reach
the scalar superhedging price, that is the portfolio $\of{25, 0}^T$, by selling the stock
at initial time. This is illustrated in figure~\ref{ex_toy}. Therefore, the portfolio $(-80,5)^T$ does not allow to superreplicate if initial positions are "cash only" positions.
Consequently, if initial positions in several eligible assets (here stock and cash) are
allowed instead of only one (like cash) the cost of superreplication can be reduced:
The portfolio $(-80,5)^T \in SHP_0(X)$ clearly allows to superreplicate (in this example even to replicate) the
claim. The scalar superhedging price gives information about price bounds, but for the purpose of actually carrying out a superhedging strategy, only the set of all superhedging portfolios gives full information on optimal strategies.
\end{example}
\begin{figure}[hpt]
\center
\input{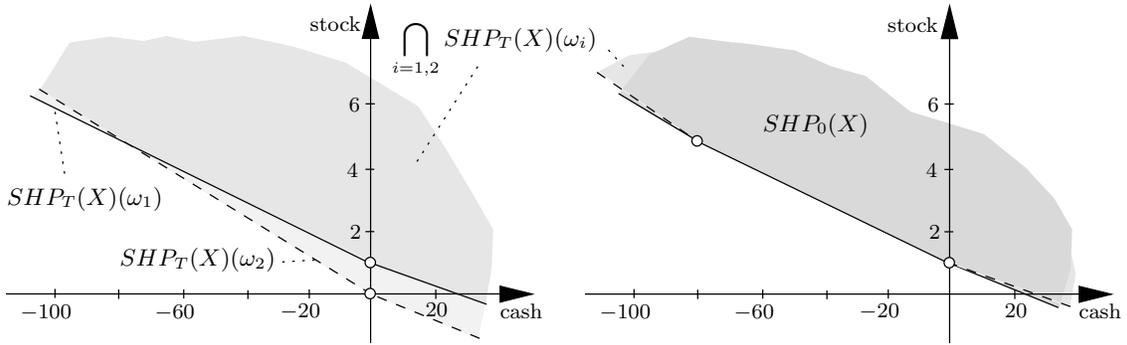}
\caption{Example~\ref{example d=2, n=2}: Illustration of the recursive algorithm \eqref{terminal condition SHP}, \eqref{eqrecursive} of theorem~\ref{theorem recursive SHP}.}
\label{ex_toy3}
\end{figure}

\begin{figure}[hpt]
\center
\input{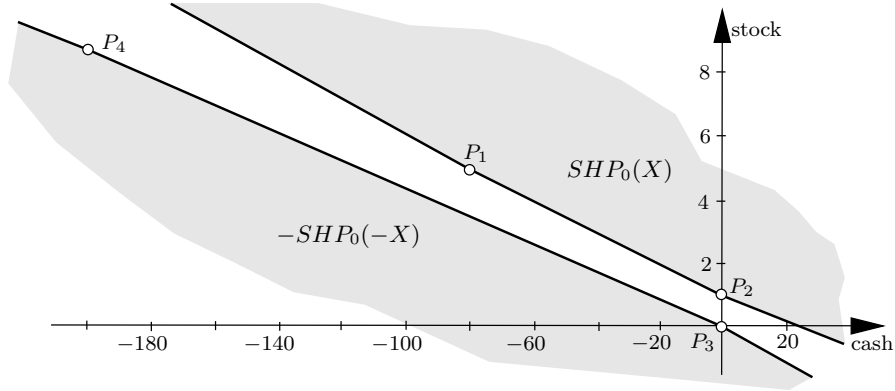}
\caption{Example~\ref{example d=2, n=2}: The set $-SHP_0(-X)$ of subhedging portfolios and the set $SHP_0(X)$ of superhedging portfolios.}
\label{ex_toy2}
\end{figure}

\begin{figure}[hpt]
\center
\input{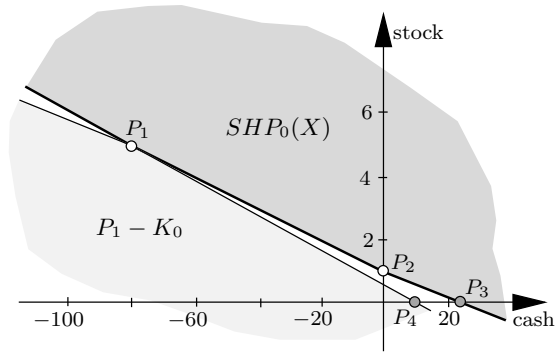}
\caption{Example~\ref{example d=2, n=2}: $P_1$ can be exchanged into $P_4=(10,0)^T$, but not into the scalar superhedging price $P_3=(25,0)^T$.}
\label{ex_toy}
\end{figure}

\section{Connection to linear vector optimization and algorithms to compute the recursive representation}
\label{sec_reformulationLVOP}

We will show that the recursive problem \eqref{terminal condition SHP}, \eqref{eqrecursive} can be formulated as a sequence of linear vector optimization problems, which can be solved by a generalization of Benson's algorithm as provided in \cite{HamLoeRud13}.
Furthermore, existing algorithms \cite{Roux08,RouxZastawniak08} for the scalar superhedging price as well as their generalization to the case of more than two assets are related to this sequence of vector optimization problems via vectorial duality \cite{HeydeLoehne08,Loehne11,HamLoeRud13}, see section \ref{sec_Scalarization}.

Consider a linear vector optimization problem with a $q$-dimensional objective function. The image space $\R^q$ is partially ordered by a polyhedral convex cone $C \subseteq \R^q$ that contains no lines and has non-empty interior. For $y,z \in \R^q$ we write $y \leq_C z$, or shortly $y \leq z$, if $z-y \in C$.
We consider the problem to
\begin{equation}\tag{P}\label{p}
 \text{ minimize } P:\R^d\to\R^q \text{ with respect to } \le_C \text{ subject to } Bx \geq b,
\end{equation}
where $B \in \R^{m\times d}$, $b \in \R^m$ and $P$ is linear, i.e. $P \in \R^{q \times d}$. The feasible set of \eqref{p} is $S:=\{x\in \R^d\colon B x \ge b\}$.
The dual problem to \eqref{p} \cite{HeydeLoehne08, Loehne11} is
\begin{equation*}\label{d}
\tag{D$^*$} \text{ maximize } D^*: \R^m \times \R^q \to \R^q \text{ with respect to } \leq_K \text{ over } T,
\end{equation*}
with (linear) objective function
\[ D^*:\R^m\times\R^q \to \R^q,\quad D^*(u,w):=\of{w_1,...,w_{q-1},b^T u}^T,\]
ordering cone $K:=\R_+ \cdot (0,0,\dots,0,1)^T$,
and feasible set
\[ T:=\cb{(u,w)\in \R^m \times \R^q \colon u \geq 0,\; B^T u = P^T w,\; c^T w = 1,\; w \in C^+},\]
where $c$ is a fixed vector in $\Int C$ and $C^+:=\cb{w\in \R^q\colon \forall y\in C:\; w^T y \geq 0}$ is the dual cone of $C$. Benson's algorithm \cite{Benson98a,EhrLoeSha11,Loehne11,HamLoeRud13} 
can be used to compute solutions to both the primal and the dual problem. 

In each iteration step of \eqref{terminal condition SHP}, \eqref{eqrecursive} a linear vector optimization problem needs to be solved which in turn means, that the calculation of the set $SHP_0(X)$ means solving a sequence of linear vector optimization problems.
The objective function $P$ at time $t$ is the liquidation map, which corresponds to an exchange of those assets with no transaction costs at time $t$ into a single asset.
For instance, if there are no transaction costs between assets $i$ and $j$ at time $t$, the matrix
\begin{equation}
\label{liquidation map}
    P_{ji}= \of{e^1,\dots,e^{i-1},e^i+\pi^{ij}e^j ,e^{i+1},\dots,e^{j-1},e^{j+1},\dots,e^d}^T
\end{equation}
is considered, where $\pi^{ij}$ is the exchange rate, the price of asset $j$ in terms of asset $i$ at time $t$. If $x$ is a portfolio with $d$ assets then $P_{ji}x$ is a portfolio with $d-1$ assets which is obtained from $x$ by exchanging asset $j$ in asset $i$ without transaction costs. The same procedure is applied to the new portfolio if there are further pairs of assets having no transaction costs. This process determines the liquidation map $P$ for the solvency cone $K_t(\omega)$, $\omega\in\Omega_t$. At the end we have transaction costs between any two assets. Let us denote the matrix containing the (finite number of) generating vectors of $K_t(\omega)^+$ as columns by $\widetilde K_t(\omega)^+$ for $\omega\in\Omega$ and $t\in\{0,...,T\}$. The generating vectors of $K_t(\omega)^+$ can be obtained by vertex enumeration from the generating vectors of $K_t(\omega)$, or by the method described in \cite{LoeRud13}. 
The notion $B=\{B_i, i\in I\}$ stands for the matrix $B\in \R^{\sum_{i\in I}m_i\times d}$ containing the rows of all the matrices $B_i \in \R^{m_i\times d}$, $i\in I$, for some index set $I$.

\begin{theorem}
\label{algo}
For $\omega \in \Omega_T$ it holds $SHP_T(X)(\omega)=\cb{x \in \R^d \colon B^\omega_T x \geq b^\omega_T}$ with $B^{\omega}_T=(\widetilde K_T(\omega)^+)^T$ and $b^{\omega}_T=B^\omega \cdot X(\omega)$.

For each $t$ from $T-1$ down to $0$ and each $\omega\in\Omega_{t}$ consider the linear vector optimization problem given by
\begin{align*}
P&=\text{LiquidationMap}(K_t(\omega)),\\
B&=\cb{B^{\bar \omega}_{t+1} \colon \bar \omega \in \suc(\omega)},\\
b&=\cb{b^{\bar \omega}_{t+1} \colon \bar \omega \in \suc(\omega)},
\end{align*}
and ordering cone $C=P \cdot K_t(\omega))$.
Let $\cb{(u^1,w^1),\dots,(u^k,w^k)}$be a solution to the dual problem \eqref{d} of the above linear vector optimization problem, and set
$B^\omega_t=(P^T w^1 ,\dots  ,P^T w^k)^T$ and $b^\omega_t= (b^T u^1 , \dots ,b^T u^k)^T$. Then,
\[
SHP_t(X)(\omega)=\cb{x \in \R^d \colon B^\omega_t x \geq b^\omega_t}.
\]
\end{theorem}

\begin{proof}
The inequality representation of $SHP_T(X)(\omega)$ is easy to see. Given an inequality representation of $SHP_{t+1}(X)(\omega)$, $t\in\{0,...,T-1\}$, the recursive form \eqref{terminal condition SHP}, \eqref{eqrecursive} of theorem~\ref{theorem recursive SHP} ensures that $SHP_{t}(X)(\omega)$ can be written as the back transformation (w.r.t. the liquidation map) of the upper image $P[S]+C$, where $S:=\cb{x\colon Bx\geq b}$, of the above vector optimization problem. Strong duality (see 
\cite{HamLoeRud13})
between the primal problem \eqref{p} and its dual problem \eqref{d} is satisfied since the feasible sets $S$ of \eqref{p} and $T$ of \eqref{d} are non-empty, which follows from \eqref{autokaputt}. Strong duality ensures that a solution to the dual problem \eqref{d} leads to an inequality representation of the upper image $P[S]+C$ of the primal problem \eqref{p}, and, using the back transformation of the liquidation map, to an inequality representation of  $SHP_t(X)(\omega)$ as given above, see \cite{HamLoeRud13,Heyde11}.
\end{proof}
The liquidation map is in place to ensure that the ordering cone $P \cdot K_t(\omega)$ contains no lines, which in turn makes it possible to apply Benson's algorithm.
If the bid-ask spread is strictly positive between any two assets at time $t$, then the liquidation map $P$ reduces to the identity and
$SHP_{t}(X)(\omega)$ is just
\[
SHP_{t}(X)(\omega)=\{x: Bx\geq b\}+K_t(\omega).
\]
In this case we deal with a very special linear vector optimization problem, which could be solved by vertex enumeration only. Recall that also in the general case, the recursive structure \eqref{terminal condition SHP}, \eqref{eqrecursive} consists of the following operations: intersections of  polyhedra and the sum of polyhedral sets and polyhedral cones, which could be realized by methods from computational geometry, essentially based on the vertex enumeration problem, see e.g. \cite{BarDobHuh96,BreFukMar98}. This direct calculation is also proposed in \cite{RouxZastawniak11} in a more general framework.
However, in practice one has to deal with numerical inaccuracy and tractability of the problem. The reformulation as a sequence of linear vector optimization problems and the usage of Benson's algorithm can have numerical advantages in both cases,
 since the polyhedra at intermediate steps of the algorithm can be approximated with error bounds that can be chosen, see remark~4.10 in~\cite{HamLoeRud13}. Without this approximation the number of vertices can grow rapidly and make the problem numerically intractable. Note that Benson's algorithm yields both outer and inner approximations, which allows to control the approximation error.
 
To illustrate this, let us assume for the moment that $P$ is the identity. Benson's algorithm in~\cite{HamLoeRud13} with pre-specified error level $\epsilon>0$ and parameter $c\in\Int K_t(\omega)$ applied to the linear vector optimization problem given in theorem~\ref{algo} leads to an inner and outer $\epsilon$-approximation of the set $SHP_{t}(X)(\omega)$ by providing a set $\overline{SHP}_{t}(X)(\omega)$ satisfying 
\[
	\overline{SHP}_{t}(X)(\omega)-\epsilon\{c\}\supseteq SHP_{t}(X)(\omega)\supseteq \overline{SHP}_{t}(X)(\omega).
\] 
 It is also important to study the propagation and accumulation of the approximation errors over time as now at time $t-1$ instead of the true input $SHP_{t}(X)(\omega$) the inner $\epsilon$-approximation $\overline{SHP}_{t}(X)(\omega)$ is used, and at each time point a new approximation error might be made. However, the following observations can be made, see \cite{FR13Bellman}: An $\epsilon$-error in the input leads to an $\epsilon$-approximation of the output providing it was calculated exactly; if the output is only calculated with approximation error $\delta>0$, the total approximation error is $\epsilon+\delta$. That means, overall one is able to calculate at time $t=0$ an $\epsilon\cdot T$-approximation of $SHP_0(X)$, that is a set $\overline{SHP}_{0}(X)$ satisfying 
\[
	\overline{SHP}_{0}(X)-\epsilon\cdot T\{c\}\supseteq SHP_{0}(X)\supseteq \overline{SHP}_{0}(X),
\] 
when using the same error level $\epsilon>0$ and parameter $c\in\Int \R^d_+$ for all times $t$ and states $\omega\in\Omega_t$.

\section{Examples}
\label{sec_Examples}
If one is interested in the superhedging portfolios for initial positions in only a few of the $d$ assets, for example in just a few currencies, or even in just one currency, one would calculate
\begin{equation}
\label{SHPM}
	SHP_0^M(X)=SHP_0(X)\cap M,
\end{equation}
for $M=\{\sum_{i\in I} s_i e^i, s_i\in\R\}$, where the assets $i\in I$ with $I\subseteq\{1,...,d\}$ are the ones of interest. This simply involves one more operation (the intersection) in the algorithm. The introduction of $M$ is quite useful if the number of assets $d$ is very high, one owns just a few of the $d$ assets and thus is interested in superhedging portfolios just starting with those assets, or if $SHP_0(X)$ is too complex to be visualized. In particular, the $i$-th component of a vertex of $SHP_0^M(X)$ for $M=\{se^i, s\in\R\}$ coincides with the smallest superhedging prices if asset $i\in\{1,...,d\}$ is chosen as the num\'{e}raire, thus coinciding with the scalar superhedging price $\pi^a_i(X)$.

A second possibility to calculate the scalar superhedging price $\pi^a_i(X)$ w.r.t. a  num\'{e}raire asset $i$ is by normalizing the inequality representation $SHP_0(X)=\{x \in \R^d \colon B x \geq b\}$ obtained by the algorithm in section~\ref{sec_reformulationLVOP} in the following way: Transform $B x \geq b$ into $\widetilde{B} x \geq \widetilde{b}$ such that each element in the $i$th row of $\widetilde{B}$ is equal to $1$. This is always possible by \eqref{eq_rec} and $SHP_0(X)\neq \R^d$. Then, the largest component of $\widetilde{b}$ is $\pi^a_i(X)$ w.r.t. the  num\'{e}raire asset $i$. This relation also plays a role in section~\ref{subsec_scalar problem as LVOP}.
Both methods provide a way to compare our results with the algorithm for the $d=2$ case of \cite{RouxZastawniak08,Roux08} for calculating the scalar superhedging price when the num\'{e}raire asset is the riskless asset (see example \ref{example d=2, multiperiod}).

\subsection{Two asset case}

\begin{example}
\label{example d=2, multiperiod, digital option}
Let us consider a digital option similar to example~\ref{example d=2, n=2}, but in a multi-period framework and smaller transaction costs.
Let asset $1$ be a riskless bond $B$ with an annual interest rate of $3\%$, face value $B_T=1$, maturity one year, frequent compounding with $n=100$ time intervals, i.e. $B_0=(1+\frac{r}{n})^{-n}$ and no transaction costs for the bond, i.e. $B_t^b=B_t^a=B_t$ for all $t$.
Let the mid-market stock price $S$ follow a Cox-Ross-Rubinstein binomial model,
\[
    S_t=\varepsilon_tS_{t-1},
\]
for $t=1,...,T$, where $\varepsilon_1,\varepsilon_2,...$ is a sequence of independent identically distributed random variables taking two possible values $e^{\sigma\Delta t}$ or $e^{-\sigma\Delta t}$, where $\Delta t$ is the length of one time step. The initial stock price is $S_0 = 18$, volatility $\sigma =0.2$, and transaction costs are constant $\lambda = 0.04$. Let the bid and the ask price at time $t$, respectively, be given by
\begin{equation}
\label{example bidask}
    S_t^b=S_t(1-\lambda),\quad\quad\quad S_t^a=S_t(1+\lambda).
\end{equation}
An asset or nothing call option with physical delivery, maturity $1$ year, strike $K=19$ and
payoff
\[
X\of{\omega} = \of{X_1\of{\omega}, X_2\of{\omega}}^T = \of{0, I_{\cb{S_T^a \geq K}}\of{\omega}}^T
\]
is considered. The set $SHP_0(X)$ has two vertices, one at $(0,1)^T$ and one at $(-24.92,2.39)^T$ and a recession cone equal to the solvency cone $K_0$ at initial time which is generated by $\of{-S_0^b/B_0, 1}^T$ and $\of{S_0^a/B_0, -1}^T$.

The scalar superhedging price in the num\'{e}raire asset (the bond) is $\pi^a_0(X)=19.29$ units of the bond, which corresponds to a scalar superhedging price in the domestic currency of $\pi^a(X)=18.72=S_0^a$ and the corresponding strategy is the buy and hold strategy. Note that in contrast to the trivial strategy one obtains for the scalar superhedging, the superhedging strategy can be more involved, when the initial portfolio vector is not cash-only.
\end{example}

\begin{example}
\label{example d=2, multiperiod}
Let asset $1$ be a riskless bond $B$ with an effective interest rate of $r_e=10\%$, frequent compounding, face value $B_T=1$, maturity $1$ year, i.e. $B_0=(1+r_e)^{-1}$, and no transaction costs for the bond, i.e. $B_t^b=B_t^a=B_t$ for all $t$.
Let the stock price $S$ follow a Cox-Ross-Rubinstein binomial model as in example~\ref{example d=2, multiperiod, digital option}. The initial stock price is $S_0 = 100$, volatility $\sigma =0.2$, maturity $1$ year and transaction costs are constant $\lambda = 0.00125$. Let the bid and ask prices at time $t$ be given as in \eqref{example bidask}. Consider a call option with maturity $1$ year, physical delivery and strike $K=80$ whose payoff is a function of the mid-market price, i.e.
\[
X\of{\omega} = \of{X_1\of{\omega}, X_2\of{\omega}}^T = \of{-KI_{\cb{S_T > K}}\of{\omega}, I_{\cb{S_T > K}}\of{\omega}}^T.
\]
The set of superhedging and subhedging portfolios is given by its vertices and recession cones.
For different values of $n$, the vertices (in units of bond and stock) are recorded in table~\ref{tabelle}.
The recession cone of $SHP_0(X)$ is always $K_0$, generated by $\of{-S_0^b/B_0, 1}^T$ and $\of{S_0^a/B_0, -1}^T$, whereas the recession cone of $-SHP_0(-X)$ is $-K_0$.
The scalar price bounds $\pi^b(X),\pi^a(X)$ in the domestic currency are also recorded in table~\ref{tabelle}.

For comparison purpose, we also give the scalar price bounds $\pi^b(X),\pi^a(X)$ if there are no transaction costs at $t=0$ as considered in \cite{RouxZastawniak08,Roux08,BoyleVorst92,Palmer01}. We are able to replicate the scalar results as given in table~1 of \cite{RouxZastawniak08} and table~3.1 and 3.2 of \cite{Roux08}, where the different values of the parameters $K$ and $\lambda$ are $K\in\{80,90,100,110,120\}$ and $\lambda\in\{0\%,0.125\%,0.5\%,0.75\%,2\%\}$. Let $n$ be the number of time intervals (that is $n=T$).
Minor deviations (all less than $0.001$) from table~3.2 of \cite{Roux08} appear in a few instances for the bid prices in the $n=1000$ case and one deviation of $0.014$ for $n=1000$ that is recorded in table~\ref{tabelle}.
The case $n=1800$ was not considered in \cite{RouxZastawniak08,Roux08}. Here, we just present the results for $K=80$ and $\lambda=0.125\%$. We used the dual variant of Benson's algorithm in \cite{HamLoeRud13} with a precision of $\epsilon=10^{-7}$, which according to the remarks at the end of section~\ref{sec_reformulationLVOP} leads to an overall precision of $\epsilon n$.
\begin{table}[h!]\scriptsize
\begin{tabular}{|c|c|c|c|c|c|c|}
            \hline
            \multicolumn{7}{|c|}{$\lambda=0.125\%$ for all $t$} \\\hline
            $n$ & $6$ & $13$ & $52$ & $250$ & $1000$ & $1800$\\\hline
             &  &  &  &  &  & \\
            vertex of $-SHP_0(-X)$&  $\left(
                                      \begin{array}{r}
                                        -74.434 \\
                                        0.953 \\
                                      \end{array}
                                    \right)$ & $\left(
                                                    \begin{array}{r}
                                                               -74.699 \\
                                                               0.956 \\
                                                             \end{array}
                                                           \right)$  & $\left(
                                                                            \begin{array}{r}
                                                                                       -75.477
                                                                                       \\
                                                                                       0.962 \\
                                                                                     \end{array}
                                                                                   \right)$&  $\left(
                                                                                                \begin{array}{r}
                                                                                                  -76.348
                                                                                                  \\
                                                                                                  0.969 \\
                                                                                                \end{array}
                                                                                              \right)$&  $\left(
                                                                                                \begin{array}{r}
                                                                                                   -78.049
                                                                                                    \\
                                                                                                   0.983 \\
                                                                                                \end{array}
                                                                                              \right)$
                                                                                              &  $\left(
                                                                                                \begin{array}{r}
                                                                                                   -79.049
                                                                                                   \\
                                                                                                   0.992 \\
                                                                                                \end{array}
                                                                                              \right)$
                                                                                   \\
&  &  &  &  & &\\
            lower price bound $\pi^b(X)$ & $27.552$ &  $27.537$ & $27.462$ & $27.381$ & $27.249$ & $27.191$\\
            &  &  &  &  &  &\\
            vertex of $SHP_0(X)$&  $\left(
                                    \begin{array}{r}
                                      -73.814 \\
                                      0.948 \\
                                    \end{array}
                                  \right)$ & $\left(
                                                \begin{array}{r}
                                                  -73.857 \\
                                                  0.949 \\
                                                \end{array}
                                              \right)$  & $\left(
                                                             \begin{array}{r}
                                                               -73.857 \\
                                                               0.949 \\
                                                             \end{array}
                                                           \right)$  &  $\left(
                                                                           \begin{array}{r}
                                                                             -72.856 \\
                                                                             0.941 \\
                                                                           \end{array}
                                                                         \right)$ &  $\left(
                                                                                                \begin{array}{r}
                                                                                                   -71.244
                                                                                                  \\
                                                                                                   0.929 \\
                                                                                                \end{array}
                                                                                              \right)$
                                                                                              &  $\left(
                                                                                                \begin{array}{r}
                                                                                                   -70.209
                                                                                                  \\
                                                                                                   0.921 \\
                                                                                                \end{array}
                                                                                              \right)$
                                                                                              \\
&  &  &  &  & & \\
            upper price bound  $\pi^a(X)$ &  $27.854$ & $27.866$ & $27.872$ & $27.994$ & $28.213$ & $28.370$ \\\hline
            \multicolumn{7}{|c|}{} \\
            \multicolumn{7}{|c|}{$\lambda=0.125\%$ for $t=1,...,T$, but no transaction costs at $t=0$ as in \cite{RouxZastawniak08,Roux08,BoyleVorst92,Palmer01}}\\\hline
            $n$ & $6$ & $13$ & $52$ & $250$ & $1000$ & $1800$ \\\hline
            &  &  &  &  &  & \\
            lower price bound  $\pi^b(X)$ & $27.671$ & $27.656$  & $27.582$  & $27.502$ & $27.372^a$ & $27.315^b$ \\
            upper price bound  $\pi^a(X)$ & $27.735$ & $27.747$  &  $27.753$ & $27.876$ & $28.097$ & $28.255^b$ \\
            \hline
            \multicolumn{7}{|c|}{} \\
            \multicolumn{7}{|l|}{$^a$ differs from value $27.386$ in \cite{Roux08}}\\
            \multicolumn{7}{|l|}{$^b$ not considered in \cite{RouxZastawniak08,Roux08}}\\
            \hline
          \end{tabular}
          \caption{set-valued and scalar sub- and superhedging portfolios of European call options}
          \label{tabelle}
\end{table}

Note that, if the bond is chosen as the num\'{e}raire asset, the scalar superhedging price $\pi^a_1(X)$ is given in units of the bond, and one needs to multiply it by $B_0$ to obtain the scalar superhedging price $\pi^a(X)$ in the domestic currency that is recorded in table~\ref{tabelle}. It is worth pointing out that there are parameter constellations that lead to multiple vertices for the set of superhedging or subhedging portfolios. For example, $-SHP_0(-X)$ for $\lambda=2\%$, $K=110$ and $n=52$ has $8$ vertices given by the columns of the following matrix
\[\left(
  \begin{array}{rrrrrrrr}
    -34.743 & -48.097 & -79.757 & -88.323 & -91.778 & -84.331 & -54.520 & -41.461 \\
    0.322 & 0.445 & 0.732 & 0.809 & 0.840 & 0.774 & 0.504 & 0.384 \\
  \end{array}
\right)
\]
with a scalar subhedging price of  $\pi^b(X)=-0.023$ (in the domestic currency) if transaction costs are considered at all time points, and $\pi^b(X)=0.865$ if no transaction costs are considered at $t=0$.
The set $-SHP_0(-X)$ for $\lambda=2\%$, $K=110$ and $n=250$ has $3$ vertices given by
\[\left(
  \begin{array}{rrr}
    2.370 & -107.125 & -110.107 \\
    -0.036 & 0.974 & 1.001 \\
  \end{array}
\right)
\]
with a scalar subhedging price of  $\pi^b(X)=-1.546$ if transaction costs are considered at all time points, and $\pi^b(X)=-0.038$ if no transaction costs are considered at $t=0$. Note that negative bid prices might occur when physical delivery is considered in markets with transaction costs. This issue was discussed and resolved in \cite{PerrakisLefoll97}, see also remark~3.30 in \cite{Roux08}.
\end{example}

\subsection{Multiple correlated assets and basket options}
\label{subsec_corr assets}
We are interested in the set of superhedging portfolios of options involving multiple correlated assets. We will use a multi-dimensional tree that approximates a $d-1$-dimensional Black Scholes model for $d-1$ risky assets, where the stock price dynamics under the risk neutral measure $Q$ are given by
\[
dS^i_t=S^i_t(r dt+\sigma_i dW^i_t),\quad\quad i=1,...,d-1
\]
for Brownian motions $W^i$ and $W^j$ with correlation $\rho_{i,j}\in[-1,1]$ for $i\neq j$.
We will follow the method in \cite{KornMueller09} to set up a tree for the correlated risky assets by transforming the stock price process $S$ into a process $Y$ with independent components. This tree will have $2^{d-1}$ branches in each node and will be recombining with $(t+1)^{d-1}$ nodes at time $t$ with $t\in\{0,1,...,T\}$.
 Thus, a node can be identified by an index $(t,j_1,...,j_{d-1})$ for $t\in\{0,1,...,T\}$ and $1\leq j_i\leq t+1$ for all $i\in\{1,...,d-1\}$. For $d=3$ the nodes at time $t$ can be described by the indices $(j_1,j_2)$ of the elements of a matrix $M_t\in\R^{(t+1)\times(t+1)}$. The values of the process $Y$ at such a node can be obtained as follows.

Let $\Sigma$ be the covariance matrix of the log asset prices and $GG^T=\Sigma$ be the Cholesky decomposition of $\Sigma$. Let $n$ be the number of time intervals and $\Delta t$ the length of one time interval. Let us denote $\alpha=G^{-1}(r-\frac{1}{2}\sigma^2)$. The initial value of the process $Y$ is given by $Y_0=G^{-1}(X_0)$ with $X_0=(\log(S_0^1),...,\log(S_0^{d-1}))$. The value of the process $Y$ at node $(t,j_1,...,j_{d-1})$ is given by
\begin{align}
\label{YKorn}
    Y_{t}^i= Y_0^i+t\alpha_i\Delta t +(2j_i-t-2)\sqrt{\Delta t},\quad\quad i = 1,...,d-1.
\end{align}
for $t\in\{0,1,...,T\}$ and $1\leq j_i\leq t+1$. We omit the index $(j_1,...,j_{d-1})$ for $Y_{t}(j_1,...,j_{d-1})$ and hope not to cause confusion. The value of the original stock price vector $S$ at this node $(t,j_1,...,j_{d-1})$ is
\begin{align}\label{transfoKorn}
    S_{t}^i=\exp^{(GY_{t})^i},\quad \quad i=1,...,d-1.
\end{align}
Now, let us assume for simplicity that the proportional transaction costs are constant for each of the risky assets and are given by $\lambda=(\lambda^1,...,\lambda^{d-1})$. Thus the bid and ask prices at node $(t,j_1,...,j_{d-1})$ are given by
\begin{align}\label{transfoBidAsk}
    (S_{t}^b)^i=S_{t}^i(1-\lambda^i)\quad\quad (S_{t}^a)^i=S_{t}^i(1+\lambda^i),\quad \quad i=1,...,d-1.
\end{align}
Furthermore, let us assume there is a riskless asset with dynamics $(B_t)_{t=0}^T$. Transaction costs in the riskless asset can be incorporated by considering bid-ask prices $B_t^b\leq B_t^a$ for $t = 0, 1, ..., T$.
For $d=3$, if both risky assets are denoted in the domestic currency (the currency of the riskless asset), if $\lambda^1, \lambda^2>0$, and if we assume an exchange between the two risky assets can not be made directly, only via cash by selling one asset and buying the other, we obtained the following tree model for the solvency cone process $K_{t}$. At node $(t,j_1,j_2)$ the generating vectors of the solvency cone are given by $\pi^{ij}e^i-e^j$, $0\leq i,j\leq 2$ (see section~\ref{sec_Preliminaries}), i.e., by the columns of the matrix
\begin{align}\label{transfoBidAskKt}
\left(
  \begin{array}{cccccc}
    \frac{(S_{t}^a)^1}{B_t^b} & -\frac{(S_{t}^b)^1}{B_t^a} &\frac{(S_{t}^a)^2}{B_t^b} & -\frac{(S_{t}^b)^2}{B_t^a}&0&0\\
    -1 & 1 & 0 & 0&1&-1 \\
    0 & 0 & -1 & 1 & -\frac{(S_t^b)^1}{(S_t^a)^2} & \frac{(S_t^a)^1}{(S_t^b)^2}\\
  \end{array}
\right).
\end{align}
If there are no transaction costs for the riskless asset, i.e. $B_t^b=B_t^a$, the last two generating vectors
in \eqref{transfoBidAskKt} are redundant and can be omitted. 

Note, if there is an asset denoted in a currency different from the domestic currency, one needs to model the exchange rate between domestic and foreign currency as well and obtains a model with one more risky asset.
The generating vectors of solvency cones in higher dimension and without the above assumptions can be obtained analogously, see definition 1.2 in \cite{Schachermayer04}.

Also, if one of the risky assets is a currency, one would rather use a discrete approximation of a mean-reverting process than a geometric Brownian motion for this asset. In this case the model for the process $Y$ with independent components would be similar to above, but the transformation \eqref{transfoKorn} needs to be adapted to the new setting, see for example \cite{HullWhite90}.

Random proportional transaction costs for an asset can be modeled analogously, by treating them as another (correlated) risky asset. In \eqref{transfoBidAsk} one would just replace the constant $\lambda$ by the value of the stochastic process $\lambda$ at node $(t,j_1,...,j_{d-1})$.

\begin{example}\label{exampe_exchange_option}
\label{example correlated tree}(exchange option) Let us consider a European option in which at expiration, the holder can exchange one unit of asset 2 and receive one unit of asset 1.  Let asset $0$ be a riskless bond $B$ with annual interest rate $r$ under frequent compounding and face value $B_T=1$. We assume constant transaction costs $\lambda_0$ for the bond with bid and ask prices as in \eqref{transfoBidAsk}. Assets $1$ and $2$ are two correlated stocks $S^1$ and $S^2$, denoted in the same currency as the bond, with initial stock price for the first stock $S_0^1=45$, volatility $\sigma_1=0.15$, constant transaction costs $\lambda_1$ and for the second stock $S_0^2=50$, $\sigma_2=0.2$, $\lambda_2$ and correlation $\rho=0.2$ between both stocks. The tree is modeled as described in section~\ref{subsec_corr assets}. The bid and ask prices are given as in \eqref{transfoBidAsk}.

Consider an exchange option with physical delivery. The payoff is given by
\[
X\of{\omega} = \of{X_1\of{\omega}, X_2\of{\omega}, X_3\of{\omega}}^T
= \of{0, I_{\cb{S_T^{a,1} \geq S_T^{a,2}}}\of{\omega}, -I_{\cb{S_T^{a,1} \geq S_T^{a,2}}}\of{\omega}}^T.
\]
The maturity is one year. Table~\ref{tabelle exchange option3} gives the vertices of $SHP_0(X)$ in units of (bond, asset $1$, asset $2)^T$ and the scalar superhedging prices $\pi^a_0(X)$ in units of bond and $\pi^a(X)$ in the domestic currency for different values for $r$ and $\lambda=(\lambda_1,\lambda_2,\lambda_3)^T$. The recession cone of $SHP_0(X)$ is equal to $K_0$ generated by the vectors given in \eqref{transfoBidAskKt}.  We used the dual variant of Benson's algorithm in \cite{HamLoeRud13} with different precisions of at most $\epsilon=2\cdot10^{-5}$.
Note that $\pi^a(X)$ in the domestic currency can be calculated straight forward if $\lambda_0=0$ by $\pi^a(X)=\pi^a_0(X)B_0$. If $\lambda_0>0$, the scalar superhedging price $\pi^a(X)$ in the domestic currency can be calculated by adding to $SHP_0(X)$ (as a $3$ dimensional object in a four dimensional space, where the cash axis was added) the four dimensional cone $\widetilde{K}_0$, which is generated by the bid and ask prices of the bond and the stocks through the vectors $((S_t^a)^1,-1,0,...,0)^T$, $(-(S_t^b)^1,1,0,...,0)^T$, $...$, $((S_t^a)^{d-1},0,...,0,-1)^T$, $(-(S_t^b)^{d-1},0,...,0,1)^T)$, and calculating the vertex of the intersection with the cash axis. The reason is that the transaction costs for the bond might lead to the effect that there might be cheaper ways to trade cash into the set $SHP_0(X)$ than to trade the pure cash position $(\pi^a_0(X))^+ B_0^a-(\pi^a_0(X))^- B_0^b$ into the pure bond position $\pi^a_0(X)$, see the fifth example in table~\ref{tabelle exchange option3}, where $\pi^a(X)=6.988<7.011=(\pi^a_0(X))^+ B_0^a-(\pi^a_0(X))^- B_0^b$. 
\begin{table}[pht!]\scriptsize
\begin{tabular}{|c|c|c|}
            \hline
            \multicolumn{3}{|c|}{$r=0\%$, $\lambda=(0\%,2\%,4\%)^T$ }  \\\hline
             n & $4$ & $20$   \\\hline
             vertices of $SHP_0(X)$ & $\left(
                                      \begin{array}{rrrrr}
                                         -7.279 &  -1.936 &   8.263 &   9.979 &  12.359 \\
                                          0.583 &   0.518 &   0.392 &   0.372 &   0.344 \\
                                         -0.264 &  -0.312 &  -0.403 &  -0.419 &  -0.441
                                      \end{array}
                                    \right)$  &$\left(
                                                          \begin{array}{rrrrr}
                                                           -4.166  & -1.616  &  1.817  &  1.960  &  4.341\\
                                                            0.569  &  0.536  &  0.492  &  0.490  &  0.461\\
                                                           -0.287  & -0.309  & -0.338  & -0.339  & -0.360\\
                                                           \end{array}
                                                         \right)$ \\
            & & \\
            $\pi^a_0(X)=\pi^a(X)$   & $6.789$& $8.158$  \\\hline \hline
      \multicolumn{3}{|c|}{$r=5\%$, $\lambda=(0\%,2\%,4\%)^T$} \\\hline
       n & $4$ & $20$ \\\hline
            vertices of $SHP_0(X)$ & $\left(
                                      \begin{array}{rrrrr}
                                         -7.650 &  -2.032 &   8.693 &  10.497 &  12.993 \\
                                          0.583 &   0.518 &   0.392 &   0.372 &   0.343 \\
                                         -0.264 &  -0.312 &  -0.403 &  -0.419 &  -0.441
                                      \end{array}
                                    \right)$ &$\left(
                                                  \begin{array}{rrrrr}
                                                       -4.379&   -1.699&    1.909&    2.060&    4.563\\
                                                        0.569&    0.536&    0.492&    0.490&    0.461\\
                                                       -0.287&   -0.309&   -0.338&   -0.339&   -0.360\\
                                                  \end{array}
                                                \right)$   \\
            & & \\
            $\pi^a_0(X)$  (in bonds)  & $7.134$ &$8.576$ \\
            $\pi^a(X)$  (in cash) & $6.788$ &$8.158$ \\\hline\hline
            \multicolumn{3}{|c|}{ $r=0\%$, $\lambda=(0\%,0.4\%,0.1\%)^T$} \\\hline
             n & $4$ & $20$ \\\hline
            vertices of $SHP_0(X)$ & $\left(
                                      \begin{array}{rrr}
                                        -5.641 &  -3.379 &  11.477 \\
                                         0.501 &   0.475 &   0.310 \\
                                        -0.259 &  -0.281 &  -0.430 \\
                                      \end{array}
                                    \right)$ & $\left(
                                                 \begin{array}{rr}
                                                    -3.703 &  3.222 \\
                                                     0.475 &  0.400 \\
                                                    -0.274 & -0.345 \\
                                                  \end{array}
                                                \right)$  \\
&&   \\
            $\pi^a_0(X)=\pi^a(X)$  & $4.032$ & $4.042$ \\\hline
    \multicolumn{3}{|c|}{$r=5\%$, $\lambda=(0\%,0.4\%,0.1\%)^T$} \\\hline
            vertices of $SHP_0(X)$ & $\left(
                                      \begin{array}{rrr}
                                         -5.946 &  -3.528 &  12.081 \\
                                          0.501 &   0.475 &   0.310 \\
                                         -0.258 &  -0.281 &  -0.430
                                      \end{array}
                                    \right)$ & $\left(
                                                  \begin{array}{rr}
                                                    -3.892 & 3.385 \\
                                                   0.475 & 0.400 \\
                                                    -0.275 & -0.345 \\
                                                  \end{array}
                                                \right)$   \\
&  & \\
            $\pi^a_0(X)$ (in bonds)  & $4.240$ & $4.249$ \\
            $\pi^a(X)$ (in cash) & $4.034$ & $4.042$ \\\hline
    \hline
    \multicolumn{3}{|c|}{$r=5\%$, $\lambda=(1\%,2\%,4\%)^T$}  \\\hline
    n & 4 & 10  \\\hline
            vertices of $SHP_0(X)$  &$\left(
                                      \begin{array}{rrr}
                                      -7.760 &  0.000 &  13.341\\
                                       0.584 &  0.498 &   0.347\\
                                      -0.260 & -0.331 &  -0.446
                                      \end{array}
                                    \right)$
                                    &
                                    $\left(
                                    \begin{array}{rrrrr}
                                       -6.379 &  -6.150 &  -6.016 &  -5.564 &  -5.272 \\
                                        0.576 &   0.573 &   0.572 &   0.567 &   0.563 \\
                                       -0.265 &  -0.267 &  -0.268 &  -0.272 &  -0.275
                                    \end{array}\right.$ \\
                               & & \\
                               & &
                               $\begin{array}{rrrrrr}
                                         -4.852 &  -3.923 &  -2.705 &   0.000 &  4.382 \\
                                          0.559 &   0.549 &   0.536 &   0.506 &  0.457 \\
                                         -0.279 &  -0.288 &  -0.299 &  -0.324 & -0.36
                               \end{array}$\\
                               & & \\
                               & &
                               $\left.
                                    \begin{array}{rrrrr}
                                        4.432 &   5.649 &   6.659  &  7.259 &   7.481 \\
                                        0.456 &   0.443 &   0.431  &  0.424 &   0.422 \\
                                       -0.363 &  -0.373 &  -0.382  & -0.387 &  -0.389
                                    \end{array}
                                    \right)$ \\
                                    & &  \\
            $\pi^a_0(X)$ (in bonds) & $7.418$ & $8.167$ \\
            $\pi^a(X)$ (in cash)    & $6.988$ & $7.692$ \\\hline\hline
    \multicolumn{3}{|c|}{$r=5\%$, $\lambda=(0.2\%,0.4\%,0.1\%)^T$}\\\hline
    n & 4 & 10  \\\hline
            vertices of $SHP_0(X)$  & $\left(
                                       \begin{array}{rrrrr}
                                          -6.236 &  -4.237 &   0.000 &   8.230 &  12.403 \\
                                           0.507 &   0.486 &   0.441 &   0.353 &   0.308 \\
                                          -0.257 &  -0.276 &  -0.317 &  -0.394 &  -0.433
                                       \end{array}
                                     \right)$
                                     &
                                    $\left(
                                    \begin{array}{rrrrr}
                                      -6.518 &  -5.820 &  -4.454 &  -2.652 &  -0.952 \\
                                       0.507 &   0.499 &   0.485 &   0.466 &   0.447 \\
                                      -0.251 &  -0.258 &  -0.271 &  -0.288 &  -0.304
                                    \end{array}\right.$ \\
                                     & & \\
                                     & &
                                     $\begin{array}{rrrrr}
 								          0.000 &   1.131 &   2.172 &   3.047  &  4.151 \\
  										  0.437 &   0.425 &   0.414 &   0.404  &  0.393 \\
 									     -0.313 &  -0.324 &  -0.333 &  -0.342  & -0.352
                                     \end{array}$\\
                                     & & \\
                                     & &
                                     $\left.
                                     \begin{array}{rrrr}
                                         5.022 &   6.501 &    7.035 &    7.235 \\
                                         0.383 &   0.367 &    0.361 &    0.359 \\
                                        -0.360 &  -0.374 &   -0.379 &   -0.381
                                    \end{array}
                                    \right)$ \\
                                    & &  \\
            $\pi^a_0(X)$ (in bonds)   & $4.310$ & 4.318 \\
            $\pi^a(X)$ (in cash)      & $4.109$ & 4.116  \\\hline
          \end{tabular}
          \caption{The set of superhedging portfolios of an exchange option with and without transaction costs for the bond (see example~\ref{exampe_exchange_option})}
  \label{tabelle exchange option3}
\end{table}
\end{example}

\begin{example}
\label{example basket option}(outperformance option: superhedging portfolios and strategies) Let us consider an outperformance option. Let asset $0$ be a riskless cash account with zero interest rates. Assets $1$ and $2$ are two correlated stocks $S^1$ and $S^2$, denoted in the same currency as the cash account, with initial stock price for the first stock $S_0^1=50$, volatility $\sigma_1=0.15$, constant transaction costs $\lambda_1=0.2$ and for the second stock $S_0^2=45$, $\sigma_2=0.2$, $\lambda_2=0.1$ and correlation $\rho=0.2$ between both stocks. The tree is modeled as described in section~\ref{subsec_corr assets}. The bid and ask are given as in \eqref{transfoBidAsk}.
The payoff under physical delivery is given by
\begin{align*}
X\of{\omega} &= \of{X_1\of{\omega}, X_2\of{\omega}, X_3\of{\omega}}^T
\\
&= \of{-K I_{\cb{\max{\{S_T^{a,1},S_T^{a,2}\}}\geq K}}\of{\omega}, I_{\cb{S_T^{a,1} \geq S_T^{a,2} \mbox{ and } S_T^{a,1}\geq K}}\of{\omega}, I_{\cb{S_T^{a,2} > S_T^{a,1}\mbox{ and }S_T^{a,2}\geq K}}\of{\omega}}^T.
\end{align*}
Let the maturity be one year and the strike $K=47$. We will use only a small number of time intervals $n=4$ to illustrate different possibilities of choosing optimal superhedging strategies as described in section~\ref{subsec_strategy}. The set of superhedging portfolios $SHP_0(X)$ has two vertices
\[
    \left(
      \begin{array}{r}
        -27.404 \\
         0.514 \\
        0.388 \\
      \end{array}
    \right), \quad \left(
                     \begin{array}{r}
                       -34.254 \\
                       0.567 \\
                       0.480\\
                     \end{array}
                   \right)
\]
and a recession cone equal to the solvency cone $K_0$,  where we used for computations the dual variant of Benson's algorithm in \cite{HamLoeRud13} with a precision of $\epsilon=10^{-8}$. The scalar superhedging price is $\pi^a(X)=22.624$ in the domestic currency and the scalar subhedging price is  $\pi^b(X)=-8.633$. Negative bid prices might occur when physical delivery is considered in markets with transaction costs. This issue was discussed and resolved in \cite{PerrakisLefoll97}, see also remark~3.30 in \cite{Roux08}.

Let us compute superhedging strategies for a given path using the method in section~\ref{subsec_strategy}. We fix a path given by the sequence of indices $j_1=1,2,3,3,4$ and $j_2=1,1,2,3,4$ at times $t=0,1,...,4$ and an initial portfolio vector given by the first vertex $x_0=(-27.404,0.514,0.388)^T$. No trading is necessary and no withdrawal is possible at times $t=0,3,4$. At time $t=1$, trading is necessary (buy $0.167$ of $S^1$ at price $64.491\$$) and no withdrawal is possible. 
At $t=2$, $2.882\$$ can be withdrawn while still guaranteeing superhedging by buying $0.320$ of $S^1$ at price $69.319\$$ and selling  $0.388$ of $S^2$ at price $41.733\$$. This replicates the claim for this path after the total withdrawal of  $2.882\$$ at $t=2$.
\end{example}

\section{Scalar superhedging price}
\label{sec_Scalarization}
The set of superhedging portfolios plays an important role when one is actually interested in carrying out a strategy starting from an initial portfolio vector that can contain more assets than just cash. Clearly, in this case, it is not enough to know the scalar superhedging price.

On the other hand, if one is interested solely in price bounds for a claim in a certain currency, the scalar superhedging and subhedging prices $\pi^a(X)$ and $\pi^b(X)$ give exactly this information. From no arbitrage it follows that the market bid and ask prices $p^b(X)$ and $p^a(X)$ of a claim $X$ have to satisfy
\begin{align}
\label{NA bounds}
    p^b(X)\leq \pi^a(X),\quad  p^a(X)\geq \pi^b(X)\quad  \mbox{ and }\quad  p^b(X)\leq p^a(X),
\end{align}
where $p^b(X),p^a(X),\pi^b(X),\pi^a(X)$ are all denoted in the same currency. Furthermore, the inequalities
\[
    p^b(X)\geq \pi^b(X),\quad\mbox{ and } \quad p^a(X)\leq \pi^a(X)\quad
\]
are reasonable: Obviously, one would rather superreplicate $X$ with $\pi^a(X)$ than to  buy it at a higher price if $p^a(X)> \pi^a(X)$, but this last inequality would not create arbitrage as long as the inequalities in \eqref{NA bounds} are satisfied. In total, one obtains price bounds
\[
\pi^b(X)\leq p^b(X)\leq p^a(X)\leq \pi^a(X).
\]
In the following, we will discuss how $\pi^b(X)$ and $\pi^a(X)$ can be calculated and how the obtained algorithm is related to the algorithm studied in theorems~\ref{theorem recursive SHP} and \ref{algo}.  In section~\ref{subsec_scalarization of dual}, theorem~\ref{theorem JK}, a dual representation of the scalar superhedging price is given. The result is the $d$-dimensional version of Jouini, Kallal \cite{JouiniKallal95} and can be obtained by scalarizing the dual representation of the set of superhedging portfolios. The dual representation of the scalar superhedging price allows to deduce dynamic programming equations (corollary~\ref{lemma dyn progr scalar shp}) that allow to implement and efficiently calculate the scalar superhedging price of a claim (algorithm in corollary~\ref{lemma alg scalar shp} in section~\ref{section scalar algorithm}).
In section~\ref{subsec_scalar problem as LVOP} the relation between the scalar algorithm and the algorithm of theorem~\ref{theorem recursive SHP} via geometric duality is discussed. One obtains that the calculation of $\pi^a(X)$ reveals also the set $SHP_0(X)$ if a certain mapping is applied to a function appearing in the penultimate step of the scalar algorithm (lemma~\ref{lemma scalar vs set algo}). 

\subsection{Scalarization of the dual representation}
\label{subsec_scalarization of dual}
In the previous sections, the set $SHP_0(X)$ of all initial portfolio vectors that allow to superhedge a claim $X\in L^0_d(\mathcal{F}_T,\R^d)$ was studied. If one is interested in the calculation of price bounds, it is helpful to study the smallest superhedging prices (and the largest subhedging prices) in the currencies or num\'{e}raire of interest. To do so, consider the one dimensional subspace $M$ (see also \eqref{SHPM}) given by $M=\{se^i, s\in\R\}$, where asset number $i$ is the chosen num\'{e}raire. One might repeat this procedure for different currencies/num\'{e}raires if one is interested in price bounds of the portfolio $X$ in different currencies/num\'{e}raires. We focus on those superhedging elements that lie in $M$, i.e. we consider $SHP_0^M(X)$. From \eqref{primalRM} and theorem~\ref{ThmSH}, we obtain
\begin{align}
    SHP_0^M(X)&= \cb{x_0\in M \colon X \in  x_0+A_T}
    \\
    \label{dualRM M}
    &= \bigcap_{\cb{\of{Q, w} \in \mathcal{W}^1
    }
    }\of{E^Q\sqb{X} + G\of{w}}\cap M
    \\
    \label{dualRMSchachi M}
    &= \cb{ x_0\in M \colon \forall Z\in\mathbb{Z}: 
    \; E[X^TZ_T]\leq x_0^TZ_0}.
\end{align}
To obtain the smallest superhedging price, we apply the
scalarization procedure introduced in \cite{HamelHeyde10} to the function $R(X)=SHP_0^M(X)$. That is, we consider the extended real-valued function $\vp_{R, v} \colon L^0_d(\mathcal{F}_T,\R^d)
\to \R\cup\cb{\pm\infty}$ given by
\[
\vp_{R, v}\of{X} = \inf_{u \in R\of{X}} v^Tu
\]
for $v \in (K_0^M)^+$ (see also section~5.1 in \cite{HamelHeydeRudloff10}). $(K_0^M)^+$ denotes the positive dual cone of the cone $K_0^M=K_0\cap M$ in $M$. Thus,
\[
(K_0^M)^+ = \cb{v \in M \colon \forall u \in K_0^M \colon v^Tu \geq 0} \subseteq M.
\]
We will apply the scalarization to the dual representation of $SHP_0^M$ given in \eqref{dualRM M}, respectively \eqref{dualRMSchachi M}. Dynamic programming equations can be obtained. This leads to an algorithm that goes backwards in the event tree.
Recall that in our case $M=\{se^i, s\in\R\}$, where asset number $i\in\{1,...,d\}$ is the asset of interest, e.g. the USD cash account, or a bond. 
Thus, we are scalarizing with respect to $v=e^i\in (K_0^M)^+=\{se^i, s\in\R_+\}$. The calculation of the smallest superhedging price in
the asset of interest, that is the calculation of $\pi^a_i(X)$ leads to a generalization of the well known Jouini, Kallal \cite{JouiniKallal95} representation to the $d$ asset case, see also \cite{Stettner00}. For simplicity we assume that the solvency cone $K_t$ contains no lines. Then, the solvency cone $K_t(\omega)$ is spanned by the vectors $\pi^{ij}e^i-e^j$, $1\leq i,j\leq d$, see section~\ref{sec_Preliminaries}. The general case can be derived using the liquidation map in \eqref{liquidation map}.

\begin{theorem}
\label{theorem JK}
Under the no arbitrage condition (NA), the scalar superhedging price $\pi^a_i(X)$ in units of asset $i\in\{1,...,d\}$ is given by
\begin{align}
    \label{JK d assets}
   \pi^a_i(X)= \sup_{(S_t, Q)\in\mathcal Q^i} E^Q[X^T S_T],
\end{align}
where $\mathcal Q^i$ is the set of all processes $(S_t)_{t=0}^T$ with $S_t^i\equiv 1$, $S_t^k\leq \pi^{jk}S_t^j$ for all $1\leq j,k\leq d$ and their equivalent martingale measures $Q$. 
\end{theorem}

\begin{proof}
From the scalarization of \eqref{dualRMSchachi M} with respect to $v=e^i\in (K_0^M)^+$ one obtains
\begin{align}   \label{SHdual}
   \pi^a_i(X)=\inf_{u \in SHP_0^M\of{X}} v^Tu
=\min\{t\in \R \colon  \sup_{Z\in\mathbb{Z}}
   E[X^T\frac{Z_T}{Z_0^i}]\leq t\}.
\end{align}
Note that for every $Z\in\mathbb{Z}$, one can define the corresponding frictionless price of the $d$ assets expressed in asset $i$ as
\begin{equation}\label{eq_one}
    S_t=(\frac{Z_t^1}{Z_t^i}, \frac{Z_t^2}{Z_t^i},...,\frac{Z_t^d}{Z_t^i}),
\end{equation}
i.e., $S_t^i\equiv 1$,and obtains an equivalent martingale measure $Q$ of the process $(S_t)_{t=0}^T$ via
$ \frac{dQ}{dP}=\frac{Z_T^i}{Z_0^i}$.
Note that $Z_t^i>0$, $t\in\{0,...,T\}$ is ensured by assumption~\eqref{assumption K_t}. The set $\mathbb{Z}$ of all consistent price systems $Z$ is one-to-one (up to a multiplicative factor for $Z$) to the set of processes $(S_t)_{t=0}^T$ with $S_t^i\equiv 1$, $S_t^k\leq \pi^{jk}S_t^j$ for all $1\leq j,k\leq d$ and their equivalent martingale measures $Q$. This follows from the observation that $S_t$ is defined by $Z\in\mathbb{Z}$ with $Z\in K^+_t\bs\cb{0}$ via \eqref{eq_one} and $K_t$ is spanned by the vectors $\pi^{ij}e^i-e^j$, $1\leq i,j\leq d$. Furthermore, $Z\in\mathbb{Z}$ is a martingale
under $P$, which corresponds to $Q$ being a martingale measure for $(S_t)_{t=0}^T$, see \cite{Schachermayer04} p.24/25 for details. Since for every $Z\in\mathbb{Z}$ 
\[
E[X^T\frac{Z_T}{Z_0^i}]=E[X^T\frac{Z_T}{Z_T^i}\frac{Z_T^i}{Z_0^i}]=E^Q[X^T S_T],
\]
we can rewrite \eqref{SHdual} as in \eqref{JK d assets}.
\end{proof}

The supremum in \eqref{JK d assets} is attained if we replace $\mathcal Q^i$ by the enlarged set $\overline{\mathcal Q}^i$, that contains martingale measures that are not necessarily equivalent to $P$. We will write dynamic programming equations for problem \eqref{JK d assets} that will allow to efficiently calculate the scalar superhedging price of a claim $X\in L^0_d(\mathcal{F}_T,\R^d)$.
To do so, let us define the the sets of one-step transition densities as in \cite{CheriditoKupper10}
\[
    \mathcal D_t :=\{\xi\in L^1(\mathcal F_t,\R_+) : E_{t-1}[\xi] = 1\},\quad t = 1, . . . , T.
\]
Let $\xi_0=1$. Every sequence $((S_t,\xi_t))_{t=0}^T$ with $S_t^i\equiv 1$, $S_t^k\leq \pi^{jk}S_t^j$ for all $1\leq j,k\leq d$ for $t = 0, . . . , T$ and  $\xi_t\in \mathcal D_t$, $E_{t}[\xi_{t}S_{t}]=S_{t-1}$ for $t = 1, . . . , T$ defines an element $(S_t, Q)$ in  $\overline{\mathcal Q}^i$ by setting
\[
    \frac{dQ}{dP}=\xi_1\cdot...\cdot\xi_T.
\]
On the other hand, every element $(S_t, Q)\in\overline{\mathcal Q}^i$ induces a sequence with the above properties by setting
\[
    \xi^Q_t:=\left\{\begin{array}{cc}
               \frac{E_{t}[\frac{dQ}{dP}]}{E_{t-1}[\frac{dQ}{dP}]} &\mbox{on}\;\{E_{t-1}[\frac{dQ}{dP}]>0\}
               \\
               1 & \mbox{on}\;\{E_{t-1}[\frac{dQ}{dP}]=0\}.
             \end{array}\right.
\]
Let us denote for $t = 1, . . . , T$
\[
\mathcal A_t^i(S_{t-1})=\left\{(S_t,\xi_t):S_t^i=1, S_t^k\leq \pi^{jk}S_t^j, 1\leq j,k\leq d,\;\xi_t\in \mathcal D_t,\; E_{t}[\xi_{t}S_{t}]=S_{t-1}\right\}.
\]
From the considerations above, we obtain the following.

\begin{corollary}
\label{lemma dyn progr scalar shp}
Assume the no arbitrage condition (NA). The scalar superhedging price $\pi^a_i(X)$ given in \eqref{JK d assets} can be written as a sequence of nested optimization problems.
Let the value function at time $T-1$ be
\begin{align}
\label{dp initial step}
    V_{T-1}(S_{T-1})=\max_{(S_T,\xi_T)\in\mathcal A_T^i(S_{T-1})}E_{T-1}[\xi_T X^T S_T].
\end{align}
For $t\in\{T-2,...,0\}$ we define the value function
\begin{align}
\label{dp intermediate step}
    V_{t}(S_{t})=\max_{(S_{t+1},\xi_{t+1})\in\mathcal A_{t+1}^i(S_{t})}E_{t}[\xi_{t+1} V_{t+1}(S_{t+1})].
\end{align}
Then
\begin{align}
\label{dp terminal step}
    \pi^a_i(X)=\max_{S_0\in\R^d: S_0^i=1, S_0^k\leq \pi^{jk}S_0^j, 1\leq j,k\leq d}V_0(S_0).
\end{align}
\end{corollary}
\begin{proof}
This follows from the one-to-one correspondence between the set $\overline{\mathcal Q}^i$ and
$\{((S_t,\xi_t))_{t=0}^T: S_0^i=1, S_0^k\leq \pi^{jk}S_0^j, 1\leq j,k\leq d,(S_{t},\xi_{t})\in\mathcal A_t^i(S_{t-1}), t = 1, . . . , T \}$ and the tower property.
\end{proof}

\subsection{Algorithm for the scalar superhedging price}
\label{section scalar algorithm}
We will split each of the iteration steps in \eqref{dp intermediate step} into two substeps.
In the first steps, one incorporates only the constraint $S_{t+1}^i=1$, $S_{t+1}^k\leq \pi^{jk}S_{t+1}^j$ for $1\leq j,k\leq d$, the
second steps coincides with actually solving \eqref{dp intermediate step}.
For $d=2$, this algorithm coincides with algorithm~4.1 in \cite{RouxZastawniak08}, and algorithm~3.15
in \cite{Roux08}. 

\begin{corollary}{(Algorithm scalar superhedging price under assumption~(NA))}
\label{lemma alg scalar shp}
\begin{enumerate}
\item \begin{align}
\label{alg cut T}
           V_{T}(S_{T})=\widetilde{V}_{T}(S_{T}) =
           \,\left\{\begin{array}{r@{\quad:\quad}l}
           X^T S_T &  S_T^i=1, S_T^k\leq \pi^{jk}S_T^j, 1\leq j,k\leq d
           \\[0.3cm]
           -\infty &  \mbox{else.}
           \end{array}\right.
       \end{align}

  This defines a function   $V_{T}^{\omega}:\R^d\rightarrow\R\cup\{-\infty\}$ at each node $\omega\in \Omega_T$.

\item For $t\in\{T-1,...,0\}$ and nodes $\omega\in \Omega_t$

\begin{align}
\label{alg cap}
    V_{t}^{\omega}(S_{t})=\capf\{\widetilde{V}_{t+1}^{\bar\omega}(S_{t+1}): \bar\omega\in\suc(\omega)\}.
\end{align}
    \begin{align}
    \label{alg cut t}
           \widetilde{V}_{t}(S_{t}) =
           \,\left\{\begin{array}{r@{\quad:\quad}l}
           V_{t}(S_{t}) &  S_t^i=1, S_t^k\leq \pi^{jk}S_t^j, 1\leq j,k\leq d
           \\[0.3cm]
           -\infty &  \mbox{else.}
           \end{array}\right.
       \end{align}

\item The scalar superhedging price of $X\in L^0_d(\mathcal{F}_T,\R^d)$ is given by
          \[
            \pi^a_i(X)=\max_{S_0\in\R^d}\widetilde{V}_0(S_0).
           \] \end{enumerate}
\end{corollary}

The {\em concave cap function} $\capf \{f_1,...,f_m\}:\R^d\rightarrow\R\cup\{-\infty\}$ of a finite number of concave functions $f_i:\R^d\rightarrow\R\cup\{-\infty\}$, $i=1,...,m$ is defined by its hypograph via
\begin{align}\label{capset}
     \hypo(\capf \{f_1,...,f_m\})=\clco\of{\bigcup_{i=1}^m \hypo f_i},
\end{align}
compare \cite{Roux08}, p.135. The closure can be omitted if $f_i$, $i=1,...,m$ are polyhedral, i.e., the hypographs of $f_i$ are polyhedral convex sets. Equation~\eqref{capset} already indicates a relationship between the cap function and set-operations.
\begin{proof}{(proof of corollary~\ref{lemma alg scalar shp})}
    It follows from the representation of the cap function as an optimization problem as in
    lemma~A.4 in \cite{Roux08} that the two substeps \eqref{alg cut t} (respectively \eqref{alg cut T} for $t=T$) and  \eqref{alg cap}  coincide with optimization problem \eqref{dp intermediate step}.
\end{proof}

\begin{remark}
   Obviously, the scalar superhedging price at time $t$ and node $\omega\in \Omega_t$ is given by
   \[
  (\pi^a_i)_t(X)(\omega)=\max_{S_t\in\R^d}\widetilde{V}_t^{\omega}(S_t)
  =\max_{S_t\in\R^d: S_t^i=1, S_t^k\leq \pi^{jk}_t(\omega)S_t^j, 1\leq j,k\leq d}V_t^{\omega}(S_t).
   \]
\end{remark}


\subsection{Interpretation in terms of vector optimization and recovering the set of superhedging portfolios}
\label{subsec_scalar problem as LVOP}

We show in this section that the algorithm for the scalar superhedging price is closely related to the algorithm which computes the set of superhedging portfolios. This means that $SHP_0(X)$ is also obtained by the scalar algorithm and, depending on how the scalar algorithm is realized, it practically coincides with a special case of the algorithm of theorem~\ref{theorem recursive SHP}. This relationship is established using duality of the linear vector optimization reformulation. 

\begin{lemma}
 \label{lemma scalar vs set algo}
Let the assumptions of theorem~\ref{theorem recursive SHP} be satisfied and let $\widetilde V^{ \omega}_t(S_t)$ be as in corollary~\ref{lemma alg scalar shp}. For $t=0,...,T$ and $\omega\in\Omega_t$,
\begin{equation}\label{baumfaellung}
 SHP_t(X)( \omega) = \big\{x \in \R^d \colon  (S_t)^T x \geq \widetilde V^{ \omega}_t(S_t), \;
 (S_t,\widetilde V^{ \omega}_t(S_t)) \text{ vertices of } \hypo \widetilde V_t^{ \omega}\big\},
\end{equation}
\end{lemma}
\begin{proof}
We start with a reformulation of the scalar algorithm. As it can be seen from corollary~\ref{lemma dyn progr scalar shp} and \ref{lemma alg scalar shp} using the first three substeps, one obtains for $ \omega \in \Omega_{T-1}$

\begin{align}\label{scalar_alg_1}
  &\widetilde V^{\omega}_{T-1}(S_{T-1})=
     \left\{
       \begin{array}{ll}
          \displaystyle\max_{(\xi^{\bar\omega},\,S^{\bar\omega})} \displaystyle\sum_{\bar\omega\in\suc(\omega)}\xi^{\bar\omega} X(\bar \omega)^T S^{\bar\omega} & \text{ if }
           S_{T-1}^i=1, S_{T-1}^k\leq \pi^{jk}_{T-1}(\omega)S_{T-1}^j, 1\leq j,k\leq d  \\
                             -\infty  & \text{ otherwise }
       \end{array}
     \right.
\end{align}
where the maximum is taken subject to the constraints
\[ (S^{\bar\omega})^i=1 , (S^{\bar\omega})^k\leq \pi^{jk}_T(\bar\omega)(S^{\bar\omega})^j, 1\leq j,k\leq d ,\quad
   \xi^{\bar\omega} \geq 0,\quad
   \sum_{\bar\omega\in\suc( \omega)} \!\!\!\xi^{\bar\omega} = 1,\quad
   \displaystyle\sum_{\bar\omega\in\suc(\omega)}\!\!\! \xi^{\bar\omega} S^{\bar\omega} = S_{T-1}.
\]
For $t=T-2,\dots,0$ and $ \omega \in \Omega_{t}$ one has
\begin{align}\label{scalar_alg_2}
   &\widetilde V^{\omega}_{t}(S_{t})=
   \left\{\begin{array}{ll}
     \displaystyle\max_{(\xi^{\bar\omega},\,S^{\bar\omega})} \displaystyle\sum_{\bar\omega\in\suc( \omega)}\xi^{\bar\omega} \widetilde V^{\bar\omega}_{t+1}(S^{\bar\omega})
        & \text{ if } S_{t}^i=1, S_{t}^k\leq \pi^{jk}_t(\omega)S_{t}^j, 1\leq j,k\leq d  \\
      -\infty
        & \text{ otherwise, }
  \end{array}
  \right.
\end{align}
where the constraints for the maximum are
\[
  \xi^{\bar\omega} \geq 0,\quad
  \displaystyle\sum_{\bar\omega\in\suc(\omega)} \xi^{\bar\omega} = 1,\quad
  \displaystyle\sum_{\bar\omega\in\suc(\omega)} \xi^{\bar\omega} S^{\bar\omega}=S_{t}.
\]
Finally, we have
\begin{align}\label{eq_hmoll1}
    \pi^a_i(X)=\max_{S_0\in \R^d} \widetilde V_0(S_0).
\end{align}
The constraints $S_t^i=1$, $S_t^k\leq \pi^{jk}_t(\omega)S_t^j, 1\leq j,k\leq d$ can be equivalently expressed as $S_t^i=1$, $S_{t} \in K^+_{t}$. The constraints in \eqref{scalar_alg_1} concerning $(\xi^{\bar\omega},\,S^{\bar\omega})$ can be replaced by $\xi^{\bar\omega} S^{\bar\omega}=\widetilde K^+_T(\bar\omega) u^{\bar\omega}$, $u^{\bar\omega} \geq 0$, $\sum_{\bar\omega\in\suc( \omega)} \widetilde K^+_T(\bar\omega) u^{\bar\omega} = S_{T-1}$ using the fact that we set $S_{T-1}^i=1$.
From \eqref{scalar_alg_1} we obtain for $\omega \in \Omega_{T-1}$,
\begin{align}\label{scalar_alg_3a}
 & \widetilde V^{ \omega}_{T-1}(S_{T-1})=
     \left\{
       \begin{array}{ll}
          \displaystyle\max_{u^{\bar\omega}} \displaystyle\sum_{\bar\omega\in\suc(\omega)} X(\bar\omega)^T \widetilde K^+_T({\bar\omega})u^{\bar\omega}& \text{ if }
            S_{T-1} \in K^+_{T-1}(\omega),\quad S^i_{T-1} = 1  \\
                             -\infty  & \text{ otherwise, }
       \end{array}
     \right.
\end{align}
where the maximum is taken subject to
\[
 u^{\bar\omega} \geq 0,\quad  \sum_{\bar\omega\in\suc( \omega)} \widetilde K^+_T(\bar\omega) u^{\bar\omega} = S_{T-1}.
\]
Setting $B^{\bar\omega} = (\widetilde K^+_T({\bar\omega}))^T$, $b^{\bar\omega} =(\widetilde K^+_T({\bar\omega}))^T X(\bar\omega)$, $w=(S^1_{T-1},...,S^{i-1}_{T-1},S^{i+1}_{T-1},...,S^{d}_{T-1},S^{i}_{T-1})$ and
omitting the $i$-th variable $S^i_{T-1}$ of $\widetilde V_{T-1}^{ \omega}$ (which does not change the problem), we see that $\hypo \widetilde V_{T-1}^{ \omega}$ is nothing else than the lower image $\D^*$ of a dual vector optimization problem \eqref{d}, that is $\D^*=D^*[T]-K$, where the parameter vector $c$, which has to be an interior point of the ordering cone, is chosen to be the $d$-th unit vector, i.e. $c=e^d$. Note that the $i$-th component of $S_{T-1}$, corresponds to the $d$-th component of $w$. This means that $c=e^d$ corresponds to the num\'{e}raire-component.

The corresponding primal problem \eqref{p} is to
\begin{equation}\label{eq_primalT}
\text{ minimize } id: \R^d\to \R^d \text{ w.r.t. } \leq_{K_{T-1}(\omega)} \text{ s.t. } B^{\bar \omega} x \geq b^{\bar \omega}, \bar \omega \in \suc(\omega).
\end{equation}
We know by theorem~\ref{theorem recursive SHP} that $SHP_T(X)(\bar \omega)=\{x\in \R^d \colon B^{\bar\omega} x \geq b^{\bar \omega}\}$ and the upper image of problem \eqref{eq_primalT} is $\P=SHP_{T-1}(X)(\omega)$. Geometric duality \cite{HeydeLoehne08,Loehne11,HamLoeRud13} 
yields that an inequality-representation of $\P$ is given by the vertices of $\D^*$, that is,
\begin{align}\label{eq_hmoll2}
 SHP_{T-1}(X)( \omega) = \big\{x \in \R^d \colon & (S_{T-1})^T x \geq \widetilde V^{ \omega}_{T-1}(S_{T-1}), \\ \nonumber
& (S_{T-1},\widetilde V^{ \omega}_{T-1}(S_{T-1})) \text{ vertices of } \hypo \widetilde V_{T-1}^{ \omega}\big\}
\end{align}
which can be expressed by a matrix $B^{ \omega}$ and a vector $b^{ \omega}$, i.e.
\begin{equation} \label{eq_hmoll3}
SHP_{T-1}(X)(\omega) = \cb{x \in \R^d \colon B^{\omega} x \geq b^{\omega}}.
\end{equation}
At time $t \in \cb{T-2,\dots,0}$ and $\omega\in\Omega_t$, $\bar \omega \in \suc(\omega)$, we use elements
$(S_{t+1},\widetilde V_{t+1}^{\bar \omega}(S_{t+1}))$ of the graph of $\widetilde V_{t+1}^{\bar \omega}$, which can be expressed as a convex combination of vertices of $\hypo \widetilde V_{t+1}^{ \bar\omega}$. The coefficients of the convex combinations can be interpreted as variables $u^{ \bar\omega}$ of the (geometric) dual \eqref{d} of the linear vector optimization problem
\begin{equation*}
\text{ minimize } id: \R^d\to \R^d \text{ w.r.t. } \leq_{K_{t}(\omega)} \text{ s.t. } B^{\bar \omega} x \geq b^{\bar \omega}, \bar \omega \in \suc(\omega).
\end{equation*}
Thus \eqref{scalar_alg_2} can be reformulated as
\begin{align}\label{scalar_alg_3b}
   &\widetilde V^{ \omega}_{t}(S_{t})=
     \left\{
       \begin{array}{ll}
          \displaystyle\max_{u^{\bar\omega}} \displaystyle\sum_{\bar\omega\in\suc( \omega)}(b^{\bar\omega})^T u^{\bar\omega} & \text{ if }
            S_{t} \in K^+_{t}( \omega),\quad S^i_{t} = 1  \\
                             -\infty  & \text{ otherwise, }
       \end{array}
     \right.
 \end{align}
where the maximum is taken subject to the constraints
\[  u^{\bar\omega}\geq 0,\quad   \sum_{\bar\omega\in\suc( \omega)} (B^{\bar\omega})^T u^{\bar\omega} = S_{t}.
\]
We see that $\hypo \widetilde V^{\omega}_{t}$ is again the lower image $\D^*$ of the dual vector optimization problem \eqref{d} for $c=e^d$. Likewise to above, using theorem~\ref{theorem recursive SHP} and geometric duality
\cite{HeydeLoehne08, Loehne11, HamLoeRud13},
we obtain \eqref{baumfaellung}, which completes the proof.
\end{proof}

\begin{remark} Theoretically, the hypographs in the preceeding result can be calculated directly using vertex enumaration. Numerical advantages of treating those problems as linear vector optimization problems are discussed at the end of section \ref{sec_reformulationLVOP}. 
\end{remark}

\begin{remark}
Note that \eqref{scalar_alg_1} and \eqref{scalar_alg_2} are parametric linear optimization problems. For every choice of the parameter $S_t$, a linear program has to be solved. It is well known that parametric linear problems of the present type correspond to linear vector optimization problems, see e.g. \cite{Focke73}.
\end{remark}

The above considerations show that the scalar algorithm is closely related to the algorithm of theorems~\ref{theorem recursive SHP} and \ref{algo} if the last step \eqref{eq_hmoll1} is omitted. The algorithm of theorem~\ref{algo} computes $SHP_0(X)$ by a very similar construction if the parameter $c=e^d$ is chosen and the input data are transformed such that the last component refers to the num\'{e}raire asset. In contrast to the algorithm of theorem~\ref{algo}, the liquidation map does not occur in lemma~\ref{lemma scalar vs set algo} and its proof. The reason is that geometric duality is not restricted to polyhedral sets that contain no lines, but Benson's algorithm is. 

\begin{remark}
A relation between the function whose epigraph is the set $SHP_t(X)$ and the function $\widetilde{V}_t$ could also be deduced via conjugation (Legendre-Fenchel transform) in the spirit of section~4.2 in \cite{RouxZastawniak09}. A similar construction has been used to prove a geometric duality theorem for convex vector optimization problems \cite{Heyde11}.
\end{remark}


\begin{thebibliography}{}


\bibitem{BarDobHuh96}  Barber, C. B., Dobkin, D. P., Huhdanpaa, H.:
                       The quickhull algorithm for convex hulls.
                       ACM Trans. Math. Softw. 22 (4), 469-483 (1996)

\bibitem{Benson98a}  Benson, H. P.:
                     An outer approximation algorithm for generating all efficient
                     extreme points in the outcome set of a multiple objective linear
                     programming problem.
                     Journal of Global Optimization 13, 1-24 (1998)

\bibitem{BensaidLesnePagesScheinkman92} Bensaid, B., Lesne, J.-P., Pag\`{e}s, H., Scheinkman, J.:
                                        Derivative asset pricing with transaction costs.
                                        Math. Finance 2 (2), 63-86 (1992)





\bibitem{BoyleVorst92} Boyle, P., Vorst, T.:
                       Option replication in discrete time with transaction costs.
                       J. Finance 47, 271-293 (1992)

\bibitem{BreFukMar98} Bremner, D., Fukuda, K., Marzetta, A.:
                      Primal-dual methods for vertex and facet enumeration.
                      Discrete Comput. Geom. 20 (3), 333-357 (1998)

\bibitem{CheriditoKupper10} Cheridito, P., Kupper, M.:
                            Composition of time-consistent dynamic monetary risk measures in discrete time.
                            International Journal of Theoretical and Applied Finance 14 (1), 137-162  (2011)

\bibitem{Ehrgott05} Ehrgott, M.:
                    Multicriteria optimization.
                    2nd edition, Springer-Verlag, Berlin (2005)

\bibitem{EhrLoeSha11} Ehrgott, M., L{\"o}hne, A., Shao, L.:
                      A dual variant of Benson's outer approximation algorithm.
                      Journal of Global Optimization 52 (4), 757-778 (2012)

\bibitem{FeiRud13} Feinstein, Z., Rudloff, B.: Time consistency of dynamic risk measures in markets with transaction costs. Quantitative Finance 13 (9), 1473-1489 (2013)

\bibitem{FR13Bellman} Feinstein, Z., Rudloff, B.: A recursive algorithm for multivariate risk measures and a set-valued Bellman's principle. Working paper

\bibitem{FoellmerSchied04} F{\"o}llmer, H., Schied, A.:
                           Stochastic Finance. Walter de Gruyter (2004)

\bibitem{Focke73} Focke, J.:
                  Vektormaximumprobleme und parametrische Optimierung (German)
                  Math. Operationsforsch. Stat. 4, 365-369 (1973)



\bibitem{Hamel11} Hamel, A. H.:
    A Fenchel-Rockafellar duality theorem for set-valued optimization.
    Optimization 60 (8-9), 1023-1043 (2011)

\bibitem{HamelHeyde10} Hamel, A. H., Heyde, F.:
    Duality for set-valued measures of risk.
    SIAM J. on Financial Mathematics 1 (1), 66-95 (2010)


\bibitem{HamelHeydeRudloff10} Hamel, A. H., Heyde, F., Rudloff, B.:
    Set-valued risk measures for conical market models.
    Mathematics and Financial Economics 5 (1), 1-28 (2011)

\bibitem{HamLoeRud13} Hamel, A. H., L{\"o}hne, A., Rudloff, B.:
    Benson type algorithms for linear vector optimization and applications. Journal of Global Optimization, DOI: 10.1007/s10898-013-0098-2 (2013)
    
\bibitem{Heyde11} Heyde, F.: Geometric duality for convex vector optimization problems. Submitted for publication. http://arxiv.org/pdf/1109.3592v1.pdf

\bibitem{HeydeLoehne08} Heyde, F., L{\"o}hne, A.:
    Geometric duality in multiple objective linear programming.
    SIAM Journal of Optimization 19 (2), 836-845 (2008)

\bibitem{HullWhite90} Hull, J., White, A.:
    Valuing derivative securities using the explicit finite difference method.
The Journal of Financial and Quantitative Analysis 25, 87-100 (1990)

\bibitem{Jahn04} Jahn, J.:
                 Vector optimization. Theory, applications, and extensions.
                 Springer-Verlag, Berlin (2004)


\bibitem{JouiniKallal95} Jouini, E., Kallal, H.:
    Martingales and arbitrage in securities markets with transaction costs.
    J. Econom. Theory 66 (1), 178-197 (1995)

\bibitem{Kabanov99} Kabanov, Y. M.:
    Hedging and liquidation under transaction costs in currency markets.
    Finance and Stochastics 3, 237-248 (1999)

\bibitem{KabanovSafarian09} Kabanov, Y. M., Safarian, M.:
    Markets with transaction costs.
    Springer (2009)
    
\bibitem{KabanovStricker01} Kabanov, Y. M., Stricker, Ch.:
 	The Harrison-Pliska arbitrage pricing theorem under transaction costs. 
	Journal of Mathematical Economics 35 (2), 185-196 (2001)
	
\bibitem{KabanovRasonyiStricker02} Kabanov, Y. M., Rasonyi, M., Stricker, Ch.:
No-arbitrage criteria for financial markets with efficient friction.
Finance and Stochastics 6, 371-382 (2002)
    

\bibitem{KornMueller09} Korn, R., M{\"u}ller, S.:
    The decoupling approach to binomial pricing of multi-asset options.
    Journal of Computational Finance 12 (3), 1-30 (2009)

\bibitem{Loehne11} L{\"o}hne, A.:
    Vector optimization with infimum and supremum.
    Springer-Verlag, Berlin (2011)

\bibitem{LoeRud13} L{\"o}hne, A., Rudloff, B., Kriesell, M.:
    On the Dual of the Solvency Cone. Working paper

\bibitem{LoehneRudloff12OR} L{\"o}hne, A., Rudloff, B.: Superhedging strategy selection via linear vector optimization in markets with transaction costs. Working paper


\bibitem{Palmer01} Palmer, K.:
    A note on the Boyle-Vorst discrete-time option pricing model with transaction costs.
    Math. Finance 11, 357-363 (2001)


\bibitem{PerrakisLefoll97} Perrakis, S., Lefoll, J.:
    Derivative asset pricing with transaction costs: an extension.
    Computational Economics 10, 359-376 (1997)

\bibitem{Rockafellar97} Rockafellar, R. T.:
    Convex analysis.
    Princeton University Press, Princeton (1997)


\bibitem{Roux08} Roux, A.:
    Options under transaction costs:
    Algorithms for pricing and hedging of european and american options under
    proportional transaction costs and different borrowing and lending
    rates.
    VDM Verlag (2008)

\bibitem{RouxZastawniak08} Roux, A., Tokarz, K., Zastawniak, T.:
            Options under proportional transaction costs:
                    An algorithmic approach to pricing and
                    hedging.
             Acta Applicandae Mathematicae 103, 201-209 (2008)

\bibitem{RouxZastawniak09} Roux, A., Zastawniak, T.:
        American options under proportional transaction costs: Pricing, hedging and stopping algorithms
        for long and short positions.
        Acta Applicandae Mathematicae 106, 199-228 (2009)

\bibitem{RouxZastawniak11} Roux, A., Zastawniak, T.:
        American and Bermudan options in currency
markets under proportional transaction costs. Submitted for publication.
        http://arxiv.org/abs/1108.1910


\bibitem{Schachermayer04} Schachermayer, W.:
    The fundamental theorem of asset pricing under proportional transaction costs in finite discrete time.
    Math. Finance 14 (1), 19-48 (2004)



\bibitem{Stettner00} Stettner, L.:
Option pricing in discrete-time incomplete market models.
Math. Finance 10 (2), 305-321 (2000) 


\end{thebibliography}
\end{document}